\documentclass[twoside]{LNGAI}

\usepackage{LNGAImacro}

\usepackage[utf8]{inputenc}
\usepackage{graphicx}
\graphicspath{{images/}}

\usepackage{amsmath}
\usepackage{amssymb}
\usepackage{xspace}
\usepackage{fancybox}
\usepackage{enumitem}
\usepackage{makecell}
\usepackage{longtable}
\usepackage{color}
\usepackage{cancel}
\usepackage{relsize}
\usepackage{epstopdf}
\usepackage{colortbl}
\usepackage{hyperref}
\usepackage{fancybox}
\usepackage{blkarray}
\usepackage{listings}
\usepackage{enumitem}
\usepackage{makecell}
\usepackage{tikz}
\usetikzlibrary{arrows,decorations,decorations.text,decorations.markings}
\usetikzlibrary{shapes,trees,positioning}
\usepackage{comment}
\usepackage{longtable}
\usepackage{orcidlink}

\newcommand{\reservedWordTiles}[1]{\ovalbox{\ensuremath{\mathsf{#1}}\xspace}}
\newcommand{\tiles}[1]{\reservedWordTiles{#1}\xspace}

\newcommand{\Soda}{\textsc{Soda}\xspace}
\newcommand{\Tiles}{\ensuremath{\mathsf{Tiles}}\xspace}
\newcommand{\AR}{\ensuremath{\mathsf{AR}}\xspace}

\newcommand{\Acrocpolis}{ACROCPoLis\xspace}
\newcommand{\Acromagat}{AcROMAgAt\xspace}

\newcommand{\fFair}{\ensuremath{\tilde{f_{\mathsf{F}}}}\xspace}
\newcommand{\fFairTauSeqi}{\ensuremath{\tilde{f}_{\tau(\mathsf{F},O)seqi}}\xspace}
\newcommand{\fPowerSet}{\ensuremath{\mathcal{P}}\xspace}
\newcommand{\fOrderedSeq}[1]{\ensuremath{{#1}_{\prec}}\xspace}
\newcommand{\fPerm}{\ensuremath{\mathrm{Perm}}}

\newcommand{\sAgent}{\ensuremath{\mathsf{A}}\xspace}
\newcommand{\sAgentAttribute}{\ensuremath{\mathsf{A_{at}}}\xspace}
\newcommand{\sResource}{\ensuremath{\mathsf{R}}\xspace}
\newcommand{\sResourceAttribute}{\ensuremath{\mathsf{R_{at}}}\xspace}
\newcommand{\sPSOutcome}{\ensuremath{\fPowerSet(\mathsf{\sAgent \times \sResource})}\xspace}
\newcommand{\tFairScen}{\ensuremath{\mathsf{F}}\xspace}
\newcommand{\sMeasure}{\ensuremath{\mathsf{M}}\xspace}
\newcommand{\sZeroOne}{\ensuremath{\{0, 1\}\xspace}}
\newcommand{\sZeroOneInterval}{\ensuremath{[0, 1]\xspace}}
\newcommand{\sBoolean}{\ensuremath{\mathbb{B}}\xspace}
\newcommand{\sNat}{\ensuremath{\mathbb{N}_{0}}\xspace}
\newcommand{\sRational}{\ensuremath{\mathbb{Q}}\xspace}
\newcommand{\sReal}{\ensuremath{\mathbb{R}}\xspace}
\newcommand{\sInteger}{\ensuremath{\mathbb{Z}}\xspace}

\newcommand{\tilestype}[1]{{\ensuremath {\textsf{#1}}}}
\newcommand{\tilesfun}[1]{{\ensuremath {\textsf{#1}}}}

\newcommand{\orcidIDlink}[1]{\href{https://orcid.org/#1}{\orcidlink{#1} \texttt{{https://orcid.org/#1}}}}

\newcommand{\name}[1]{\xspace{\textbf{({#1})}}}

\newtheorem{notation}{Notation}[section]
\newtheorem{concept}{Concept}[section]

\definecolor{oceanblue}{rgb}{0, 0.4824, 0.6549}

\newcommand{\pipelineScale}{1}

\def\lastname{Mendez and Kampik}

\begin{document}

    \begin{frontmatter}
        \title{The AR Fairness Metamodel: A Structured Framework for Fairness Measures}

        \author{Julian~Alfredo~Mendez}\footnote{Julian~A.~Mendez~\orcidIDlink{0000-0002-7383-0529},  \href{mailto:julian.mendez@cs.umu.se}{\texttt{julian.mendez@cs.umu.se}}},
        \author{Timotheus~Kampik}\footnote{Timotheus~Kampik~\orcidIDlink{0000-0002-6458-2252}, \href{mailto:tkampik@cs.umu.se}{\texttt{tkampik@cs.umu.se}}}

        \address{Umeå University, Umeå, Sweden}

        \begin{abstract}
            This paper presents the \AR fairness metamodel, a framework designed to represent, analyze, and compare different fairness scenarios.
            The metamodel considers key elements, such as agents, resources, and their attributes, and enables the systematic definition and comparison of various fairness measures.
            We provide examples involving both discrete and continuous measures, including equality, equity, group fairness, individual fairness, the Gini index, the Theil index, Jain's fairness index, and a detailed fairness measure for Australia's Child Care Subsidy.
            We also explore relationships among group fairness, individual fairness, and envy-freeness, supported by formal proofs.
            At the conceptual modeling level, our approach builds on the \Tiles framework, which offers modular components that can be connected to capture diverse fairness definitions.
            The goal is to make \AR-based fairness definitions practical and adaptable across contexts, providing a clear way to define, compare, and evaluate them.
            An implementation of the \Tiles framework is available as an open-source tool, and can support fairness modeling and evaluation across a wide range of applications.
        \end{abstract}

    \end{frontmatter}

    \section{Introduction}
    \label{sec:intro}

    Fairness is a critical consideration across domains such as social policy, economics, and technology.
    Despite its importance, defining and evaluating fairness remains a challenge, as assessments of what is fair often vary across contexts and stakeholders.
    There is no universal definition of fairness, and even seemingly technical decisions can carry nuanced fairness implications~\cite{AlerTubella-2022}.
    In a given scenario, fairness can be addressed by specifying a \emph{fairness measure} that evaluates the distribution of resources among agents.
    Although fairness measures are inherently subjective, they must be rigorously defined in critical contexts.
    This makes it essential to systematize how such measures are specified, from abstract, subjective interpretations to concrete execution, and to support the comparison of fairness definitions.

    This paper introduces the \AR fairness metamodel, designed to represent and analyze fairness measures and scenarios.
    The metamodel extends prior research in~\cite{Mendez.Kampik.Aler.Dignum-2024-SCAI} and serves as a model of models~\cite{Weske-2019}, where each model is an instance of the metamodel.
    Specific instances can then be evaluated to verify whether they comply with a given definition of fairness.
    The metamodel addresses the challenges of defining and evaluating fairness by offering a structured approach.
    It provides an abstract representation of \emph{fairness scenarios}, incorporating key elements such as agents, resources, and attributes, essential components for assessing whether an outcome adheres to a particular fairness definition.

    We employ \Tiles~\cite{Mendez.Kampik.Aler.Dignum-2024-SCAI}, a framework designed to support the \AR fairness metamodel.
    \Tiles consists of modular blocks that can be interconnected to specify fairness definitions.
    Each block is annotated to indicate how it can be connected to other blocks within the framework.
    The combination of the \Tiles framework and the \AR fairness metamodel provides a comprehensive set of tools to model and evaluate fairness across diverse scenarios, with applications in real-world contexts.

    This paper is organized as follows.
    Section~\ref{sec:background} provides an overview of computational models of fairness.
    Section~\ref{sec:metamodel} introduces the \AR fairness metamodel and its components, including identifiers, measures, attributes, and auxiliary functions.
    Section~\ref{sec:operationalization} discusses the structure of the blocks and their graphical notation.
    Section~\ref{sec:discussion} offers a discussion of the \AR fairness metamodel and the \Tiles framework, focusing on their capabilities and limitations.
    Finally, Section~\ref{sec:conclusion} concludes with reflections and directions for future work.

    This paper revises and extends work presented at the LNGAI 2025 workshop~\cite{Mendez.Kampik-2025-LNGAI,Mendez.Kampik-2025-Specification-arxiv}.
    Notably, our extension provides:
    \begin{enumerate}
        \item a concrete example of fairness applied to Child Care Subsidy in Australia;
        \item a discussion of how group fairness and individual fairness can preserve envy-freeness;
        \item an additional fairness measure instantiated within our metamodel, specifically the Theil index~\cite{Conceicao.Ferreira-2000-Theil};
        \item a formal analysis of relationships among several continuous fairness measures that we instantiate;
        \item a detailed practical demonstration of how the \Tiles framework operates, particularly the engineering of new tiles from primitive tiles.
    \end{enumerate}

    \section{Background}
    \label{sec:background}

    The importance of fairness in artificial intelligence (AI) and subfields such as machine learning is widely recognized.
    From a modeling perspective, evaluating fairness requires the ability to identify and quantify unwanted bias, which may lead to prejudice and ultimately discrimination.

    Formalizing fairness contributes to greater transparency in achieving equitable outcomes, benefiting both individuals and the groups they represent.
    Although operationalizing fairness is challenging, efforts to formalize it and automate fairness verification~\cite{Albarghouthi-2017,Albarghouthi-2019} are highly relevant.
    Several quantifiable definitions have been proposed, reflecting legal, philosophical, and social perspectives~\cite{Dwork-2012,Hardt-2016,Joseph-2016,Kearns-2018}.
    However, differing interpretations can inadvertently harm the very groups they aim to protect~\cite{CorbettDavies-2018} or fail to account for intersectionality~\cite{Kearns-2018}.

    Two widely discussed formalizations are \emph{individual fairness} and \emph{group fairness}.
    Individual fairness requires that similar individuals, based on non-protected attributes, receive similar outcomes.
    Group fairness stipulates that protected groups should receive similar outcomes when non-protected factors are equal~\cite{Chouldechova-2017}.
    These notions can conflict~\cite{Binns-2019-IndividualGroupFairness}.
    For example, if two individuals with similar qualifications receive different outcomes solely because they belong to different protected groups, group fairness metrics such as equality of odds or equality of opportunity can be applied to address the disparity.
    In practice, reconciling these notions and managing the associated value trade-offs remains an active research challenge~\cite{DBLP:conf/innovations/DworkHPRZ12,10.1145/3461702.3462621,AlerTubella-2022}.
    Model-based methodologies, such as MBFair~\cite{Ramadan-2025-SSM}, enable the verification of software designs with respect to individual fairness.

    Operational tools for fairness assessment include IBM's AI Fairness 360~\cite{DBLP:journals/ibmrd/BellamyDHHHKLMM19} (AIF360), Microsoft's Fairlearn~\cite{bird2020fairlearn}, and Google's What-if Tool~\cite{DBLP:journals/tvcg/WexlerPBWVW20} (WIT).
    AIF360 is a comprehensive open-source Python library that provides fairness metrics and bias mitigation algorithms, operating across three stages: pre-processing, in-processing, and post-processing.
    Fairlearn includes fairness metrics, bias mitigation algorithms, and fairness dashboards for visual comparisons.
    WIT is visualization-oriented and offers a dashboard to explore counterfactuals, enabling users to ask ``What if this feature changed?''.
    However, these tools focus primarily on the operationalization of fairness measures rather than their definition and analysis.
    To address this limitation, we propose a unified metamodel that supports multiple perspectives on fairness, building on the frameworks \Acrocpolis~\cite{AlerTubella-2023} and \Acromagat~\cite{Mendez.Kampik.Aler.Dignum-2024-SCAI}.
    Our aim is to integrate diverse definitions of fairness into a coherent structure, facilitating consistent evaluation and comparison across scenarios.

    \section{Fairness Metamodel}
    \label{sec:metamodel}

    This section presents \AR, a formal metamodel from which fairness definitions can be instantiated.
    Conceptually, \AR focuses on \emph{agents}, \emph{resources}, their attributes as first-class abstractions, and the resulting \emph{outcomes}.
    The presentation of \AR is accompanied by examples that demonstrate its applicability to fairness assessment in specific scenarios, as well as to the abstract comparison of fairness measures.

    \subsection{Basic Elements}
    \label{subsec:basic}

    A fairness scenario provides the building blocks for relating agents, resources, and their attributes.
    This relation, called an \emph{outcome}, is used to evaluate whether a given scenario adheres to a defined concept of fairness.
    As a prerequisite, we assume two finite background sets: one of (unspecified) \emph{agents}, denoted by $\cal A$, and one of (unspecified) \emph{resources}, denoted by $\cal R$.
    We assume that the two sets are disjoint, i.e., ${\cal A} \cap {\cal R} = \emptyset$.
    The metamodel is defined as follows.

    \begin{definition}
        \normalfont
        \name{Fairness Scenario}
        \label{def:metamodel}
        A \emph{fairness scenario} is a tuple $\tFairScen = \langle \sAgent, \sResource, \sAgentAttribute, \sResourceAttribute \rangle$, such that:
        \begin{itemize}
            \item $\sAgent \subseteq \cal A$ and $\sResource \subseteq \cal R$; both are non-empty;
            \item every \emph{agent attribute} $\mathsf{a_{AT}} \in \sAgentAttribute$ is a function that takes an agent as input; every \emph{resource attribute} $\mathsf{r_{AT}} \in \sResourceAttribute$ is a function that takes a resource as input; the codomains of agent and resource attributes may vary and are specified upon instantiation.
        \end{itemize}
    \end{definition}

    Definitions of quantities are crucial for measuring fairness.
    \begin{notation}
        \normalfont
        \name{Sets of Quantities}
        The set $\sMeasure$ is a placeholder for a set of quantities such as the real numbers ($\sReal$), rational numbers ($\sRational$), integers ($\sInteger$), or natural numbers including 0 ($\sNat$), with operations totally defined on $\sMeasure$.
    \end{notation}

    Relevant attributes include, for example:
    \begin{enumerate}[label=\roman*)]
    \item the utility function $u : \sResource \to \sMeasure$, which returns the value of a resource, and
    \item the need function $q : \sAgent \to \sMeasure$, which returns how much of a resource an agent requires.
    \end{enumerate}
    Given a fairness scenario, we can define fairness measures.
    We denote the \emph{power set} of a set $S$ by $\fPowerSet(S)$.
    \begin{definition}
        \normalfont
        \name{Fairness Measure}
        \label{def:fairness-measure}
        Let $\tFairScen = \langle \sAgent, \sResource, \sAgentAttribute, \sResourceAttribute \rangle$ be a fairness scenario.
        An \emph{outcome} $O$ is an element $O \in \sPSOutcome$.
        When the outcome $O$ is clear from context, we may say that $a$ \emph{receives} $b$ whenever $\langle a, b \rangle \in O$.
        A \emph{fairness measure} $\fFair$ with respect to a fairness scenario $\tFairScen$ is a function $\fFair : \sPSOutcome \to \sZeroOneInterval$.

        Since $\sZeroOne$ is isomorphic to $\sBoolean = \{ false, true \}$, we particularly consider the case when $\fFair (O)$ returns only 0 or 1.
        Our interpretation is that 0 corresponds to $false$ and 1 to $true$, i.e., if $\fFair (O) = 1$, the outcome is fair, and if $\fFair (O) = 0$, it is unfair.
        Generally, we use $\{ false, true \}$ and $\{ 0, 1 \}$ interchangeably.
    \end{definition}
    We now illustrate how the definition of a fairness measure can be applied.
    \begin{notation}
        \normalfont
        \name{Functions by Extension}
        \label{def:functions_by_extension}
        For a function $f : A \to B$, we denote $f$ as a set of pairs $\langle x , y \rangle$ such that $x$ ranges over the elements of $A$ exactly once and $y = f(x)$.
        We also use Iverson bracket notation, where $ f(x) = [P(x)] $ denotes that $f(x) = 1$ if $P(x)$ holds, and $f(x) = 0$ otherwise.
    \end{notation}
    \begin{example}
        \name{Fairness Scenario}
        A group of agents $\sAgent$, namely Alice ($A$), Bob ($B$), Carol ($C$), David ($D$), Eve ($E$), and Frank ($F$), apply for a subsidy.
        Assume there are three types of resources $R_{1}$, $R_{2}$, and $R_{3}$, with utilities $u$ of 10, 20, and 30 respectively.
        The agents' needs are encoded in the function $q$, where $A$ and $D$ need 10, $B$ and $E$ need 20, and $C$ and $F$ need 30 each.
        Suppose it is considered \emph{fair} to give everyone at least one of the two best resources.
        The full instantiation of the fairness scenario and fairness measure is summarized as follows:

        \begin{itemize}
            \renewcommand\labelitemi{}
            \item $\sAgent = \{A, B, C, D, E, F\}$,
            $\sResource = \{ R_{1}, R_{2}, R_{3} \}$,
            \item $\sAgentAttribute = \{ q : \sAgent \to \sNat,$ $ q = \{$ $ \langle A, 10 \rangle,$ $ \langle B, 20 \rangle,$ $ \langle C, 30 \rangle,$ $ \langle D, 10 \rangle,$ $ \langle E, 20 \rangle,$ $ \langle F, 30 \rangle \} \}$,
            \item $\sResourceAttribute = \{ u : \sResource \to \sNat, \ u = \{ \langle R_{1}, 10 \rangle, \langle R_{2}, 20 \rangle, \langle R_{3}, 30 \rangle \} \}$,
            \item $\fFair (O) = [ \forall a \in \sAgent \ (a \text{ receives } R_{2} \text{ or } a \text{ receives } R_{3}) ]$
        \end{itemize}

        Here, $\sAgent$ and $\sResource$ contain the agents and resources respectively, $q$ specifies how much each agent needs, and $u$ specifies the utility of each resource.
        Intuitively, $\fFair (O)$ requires that every agent receives $R_{2}$ or $R_{3}$ (or both).

        Considering two different outcomes:
        $O_{1}= \{ \langle A, R_{3} \rangle,$ $ \langle B, R_{3}\rangle ,$ $ \langle C, R_{3} \rangle,$ $ \langle D, R_{3} \rangle ,$ $ \langle E, R_{3} \rangle ,$ $ \langle F, R_{3} \rangle \}$ and
        $O_{2}= \{ \langle A, R_{3} \rangle,$ $ \langle B, R_{2}\rangle ,$ $ \langle C, R_{1} \rangle,$ $ \langle D, R_{3} \rangle ,$ $ \langle E, R_{2} \rangle ,$ $ \langle F, R_{1} \rangle \}$,
        we see that $\fFair(O_{1}) = 1$: $O_{1}$ is fair; in contrast, $\fFair(O_{2}) = 0$: $O_{2}$ is unfair.
    \end{example}

    The notion of fairness applied in this example does not consider the \emph{needs} of the agents.
    Fairness measures that incorporate needs are presented further below.

    \begin{example}
        \name{Modeling Child Care Subsidy}
        \label{ex:equality-equity-australia}

        A \emph{child benefit} is a social security payment distributed to parents or guardians of children.
        Different countries apply varying rules to allocate these payments.
        For instance, Australia provides Child Care Subsidy (CCS)\footnote{\url{https://www.servicesaustralia.gov.au/child-care-subsidy} (accessed 2026-09-16)}, which helps parents and guardians cover the cost of formal child care.

        We model this scenario as an instantiation of our metamodel.
        First, we identify the elements of the fairness scenario.
        The agents $\sAgent$ are all eligible parents and guardians according to the regulations\footnote{\url{https://www.servicesaustralia.gov.au/who-can-get-child-care-subsidy?context=41186} (accessed 2026-09-16)}.
        The resources $\sResource$ are the available hours of child care subsidized by the scheme, with $\sMeasure = \sNat$ used to measure Australian dollars (AUD).
        We define a utility function $u : \sResource \to \sNat$ to quantify the value of each option, in terms of hours or their equivalent in AUD, given a fixed hourly rate.
        Parents and guardians may require child care hours, modeled by a need function $q : \sAgent \to \sNat$.
        The income function $w : \sAgent \to \sNat$ represents a family's combined annual income.
        The total amount each agent receives in $O$ is determined by the function $r_{O} : \sAgent \to \sNat$.

        In Australia, the CCS provided care to 1,447,460 children across 1,016,920 families, at an average hourly fee of 13.15 AUD, representing a total budget of 3.89 billion AUD in the last quarter of 2024 (October/December)\footnote{\url{https://www.education.gov.au/early-childhood/about/data-and-reports/quarterly-reports/child-care-subsidy-data-report-december-quarter-2024} (accessed 2026-09-16)}.
        Assume the budget is divided equally, with every participant receiving the same amount.
        If divided equally by family, this results in approximately 1,275 AUD per family per month.
        If divided by child, this corresponds to about 896 AUD per child per month.
        Families with one child would prefer the former criterion, while families with multiple children would prefer the latter.
        In practice, the Australian government allows certain vulnerable groups to receive a higher subsidy\footnote{\url{https://www.education.gov.au/early-childhood/families/first-nations-activity-test} (accessed 2026-09-16)}.

        To facilitate equitable distribution, we must consider how each family can receive sufficient funding to pay for caregivers, regardless of whether they receive more than they need.
        For example, if 896 AUD per month is sufficient to cover the cost of one child's care, distributing the budget equally according to the number of children aligns with the principle of equity.
    \end{example}

    \subsection{An Analysis of Equality and Equity}
    \label{subsec:measures}

    We now demonstrate how the fairness metamodel can be applied to formalize and compare two well-known fairness measures: \emph{equality} and \emph{equity}.
    Under equality, every agent receives exactly the same amount of resources.

    \begin{definition}
        \normalfont
        \name{Equality, Equity, and Strict Equity}
        \label{def:equality-equity}

        Let $\tFairScen = \langle \sAgent, \sResource, \sAgentAttribute, \sResourceAttribute \rangle$ be a fairness scenario, let $u \in \sResourceAttribute$ be a utility function, and let $O$ be an outcome.
        The \emph{accumulation of received resources} $r_{O} : \sAgent \to \sMeasure$ is defined as:
        \begin{equation}
            r_{O}(a) = \sum _{\langle a , b \rangle \in O} u (b) \ .
        \end{equation}
        This function sums the utility accumulated by an agent.

        \begin{itemize}
            \item The \emph{equality} fairness measure $\fFair_{eqa}$ is

            \[
                \fFair_{eqa} (O) = [\forall a, a' \in \sAgent \ (r_{O}(a) = r_{O}(a')) ].
            \]

            \item If $q \in \sAgentAttribute$ is the need function, the \emph{equity} fairness measure $\fFair_{eqi}$ is

            \[
                \fFair_{eqi} (O) = [\forall a \in \sAgent \ (r_{O}(a) \geq q(a)) ].
            \]

            \item The \emph{strict equity} fairness measure $\fFair_{seqi}$ is

            \[
                \fFair_{seqi} (O) = [\forall a \in \sAgent \ (r_{O}(a) = q(a)) ].
            \]

        \end{itemize}
    \end{definition}

    In Definition~\ref{def:equality-equity}, equity requires that each agent receives at least as much as they need.
    This notion can be strengthened so that it is violated if an agent receives more than they need.
    We call this stricter version \emph{strict equity}, which constitutes a special case of equity.

    \begin{proposition}
        \normalfont
        \name{Strict Equity Implies Equity}
        \label{lem:equity-strict-equity}
        For a fairness scenario $\tFairScen$, for every outcome $O$, the following implication holds:

        \[
            \fFair_{seqi}(O) = 1 \Rightarrow \fFair_{eqi}(O) = 1 .
        \]

    \end{proposition}

    \begin{proof}
        Let $\tFairScen = \langle \sAgent, \sResource, \sAgentAttribute, \sResourceAttribute \rangle$ be a fairness scenario, and let $q : \sAgent \to \sMeasure$ be the need function with $q \in \sAgentAttribute$.
        If $\fFair_{seqi}(O) = 1$, then by definition $\forall a \in \sAgent \ (r_{O}(a) = q(a))$.
        Consequently, it follows that $\forall a \in \sAgent \ (r_{O}(a) \geq q(a))$, and therefore $\fFair_{eqi}(O) = 1$.
    \end{proof}

    More interestingly, equality can be reduced to strict equity by stipulating that the equally distributed resources are sufficient to satisfy the agents' needs.
    In other words, if identical amounts are distributed to every agent and each requires the same amount, then it is possible to \emph{reduce} a fairness scenario to another one containing a need function that satisfies all agents.
    As a prerequisite, we fix background sets of fairness scenarios $\cal F$ and outcomes $\cal O$.

    \begin{proposition}
        \normalfont
        \name{Equality Reduced to Strict Equity}
        \label{lem:equality-strict-equity}
        For every fairness scenario $\tFairScen$ and outcome $O$, there exists a function $\tau : \cal F \times \cal O \to \cal F$, such that

        \[
            \fFair_{eqa}(O) = 1 \Longleftrightarrow \fFairTauSeqi(O) = 1.
        \]

    \end{proposition}

    \begin{proof}
        Let $\tFairScen = \langle \sAgent, \sResource, \sAgentAttribute, \sResourceAttribute \rangle$ be a fairness scenario and $O$ an outcome.
        Choose an arbitrary element $a_{0} \in \sAgent$, which is well-defined because $\sAgent$ is non-empty, and define the function $q : \sAgent \to \sMeasure$ as $q(a) := r_{O}(a_{0})$ (i.e., $q$ is fixed independently of input $a$).
        Also, let $\tau(\tFairScen, O) = \langle \sAgent, \sAgentAttribute \cup \{q\}, \sResource, \sResourceAttribute \rangle$.

        ($\Rightarrow$) Assume that $\fFair_{eqa}(O) = 1$ and choose $a, a' \in \sAgent$.
        Then, by definition $r_{O}(a) = r_{O}(a')$, specifically $r_{O}(a) = r_{O}(a_{0})$.
        By definition of $q$, it holds that $q(a) = r_{O}(a_{0})$.
        Since $a$ is arbitrary, we have $\forall a \in \sAgent$, $r_{O}(a) = q(a)$, and therefore $\fFairTauSeqi(O) = 1$.

        ($\Leftarrow$) Assume that $\fFairTauSeqi(O) = 1$ and choose $a \in \sAgent$.
        Then, by definition $r_{O}(a) = q(a)$.
        By definition of $q$, it holds that $q(a) = r_{O}(a_{0})$.
        Since $a$ is arbitrary, $\forall a, a' \in \sAgent$ we have $r_{O}(a) = r_{O}(a')$, and therefore $\fFair_{eqa}(O) = 1$.
    \end{proof}

    In summary, all agents require what an arbitrary agent receives.
    Only if all agents receive the same amount are their needs strictly met.

    \begin{example}
        \name{Equity in Child Care Subsidy}
        \label{ex:equity-in-ccs}
        Building on the Australian child care subsidy scenario presented in Example~\ref{ex:equality-equity-australia}, we observe that strict equity is more nuanced than equality, since every family must receive exactly what is needed to pay for child care.
        The complexity arises because caregivers may charge differently depending on location, qualifications, availability, or by grouping multiple children within the same household.
        The Australian Government adopts an equity-based approach, as the amount provided varies according to the family's combined income\footnote{\url{https://www.servicesaustralia.gov.au/your-income-can-affect-child-care-subsidy?context=41186} (accessed 2026-09-16)}, summarized in Table~\ref{tab:ccs-according-to-income}.
        A payment of 100\% corresponds to 400 AUD per week.
        Considering the combined annual income $w$, which is an agent attribute, the need function $q$ can be defined as:

        \[
            q (a) = 400 \cdot \max \left(\frac{90 - \frac{\max(w(a) - 83280 , 0)}{5000}}{100}, \ 0 \right)
        \]

        Families are expected to pay the remaining amount per week, summing to 420 AUD.
        These values apply to partnered families with a single child\footnote{\url{https://www.startingblocks.gov.au/child-care-subsidy-calculator} (accessed 2026-09-16)}.
        A family can receive at most 360 AUD per week, equivalent to 1,543 AUD per month, which contrasts with the equality cases of 896 AUD per child or 1,275 AUD per family.

        \begin{longtable}{|p{40mm}|p{40mm}|}
            \caption{Percentage of the Australian Child Care Subsidy according to family income.}
            \label{tab:ccs-according-to-income} \\
            \hline
            \textbf{Family income} & \textbf{Child Care Subsidy percentage} \\
            \hline
            0 to 83,280 AUD        & 90\%                                   \\
            \hline
            More than 83,280 AUD to below 533,280 AUD & Between 90\% and 0\%.
            The percentage decreases by 1 point for every 5,000 AUD earned above 83,280 AUD \\
            \hline
            533,280 AUD or more    & 0\%                                    \\
            \hline
        \end{longtable}

        Under equality, every agent receives the same amount of resources, regardless of individual circumstances.
        Under equity, each agent receives at least the amount they need, independent of what others receive.
        Under strict equity, each agent receives \emph{exactly} what they need.
        These measures are defined in Definition~\ref{def:equality-equity}.
    \end{example}

    \subsection{Preferences}
    \label{subsec:preferences}

    In the examples above, we used quantitative functions to evaluate fairness in the distribution of resources among agents.
    However, qualitative functions can also serve this purpose.
    Attributes may be used to determine preferences and thus define fairness measures.
    Note that qualitative functions can also be used to group agents into categories.

    \begin{definition}
        \normalfont
        \name{Preference Function}
        \label{def:preference}
        A \emph{preference function} $v$ is a function $v : \sAgent \to \fPowerSet(\sResource \times \sResource)$, such that for each agent $a$, $v(a)$ determines a complete and transitive binary relation on $\sResource$~\cite{Osborne.Rubinstein-2020-Models}.
        For each pair $\langle b_{1}, b_{2} \rangle \in v(a)$, the interpretation is ``agent $a$ weakly prefers resource $b_{2}$ to resource $b_{1}$,'' or equivalently, ``$b_{2}$ is at least as preferred as $b_{1}$.''

        An \emph{ordinal preference function} $\sigma$ is a function $\sigma : \sAgent \to \fPerm(\sResource)$, where $\fPerm(\sResource) = \{ (b_{1}, b_{2}, \ldots, b_{n}) \mid \{b_{1}, b_{2}, \ldots, b_{n}\} = \sResource\}$, such that for each agent it defines a ranking (a strict total preference order) over all resources, ordered from most preferred to least preferred.
        The notation $b_{1} \succ_{a} b_{2}$ indicates that $b_{1}$ precedes $b_{2}$ in $\sigma(a)$, and is read as ``$a$ prefers $b_{1}$ over $b_{2}$.''

        This particular representation does not allow for \emph{ties}, i.e., situations in which an agent is indifferent between two resources.
        To accommodate ties, the definition of ranking can be relaxed to a weak order, where every pair is comparable but some pairs may be considered equally good.

        A \emph{preference score} for agents $\sAgent$ in resources $\sResource$ is a function $u : \sAgent \times \sResource \to [0, 1]$, such that each agent $a$ assigns a utility value to each resource $r$, placing all resources in total order.
        This score allows for ties and models the intensity of each agent's preference for each resource.
    \end{definition}

    \begin{example}
        \name{Ordinal Preferences}
        \label{ex:preferences}
        Assume that $\sAgent = \{A, $ $B, $ $C, $ $D, $ $E \}$, and that each agent can receive a jacket that is either \emph{large} ($L$) or \emph{small} ($S$), i.e., $\sResource = \{ L, S\}$.
        Each agent chooses jackets in some order of preference.
        To model these preferences, we define the attribute $v : \sAgent \to \fPerm(\sResource)$ as

        \[
            v = \{ \langle A, (S, L) \rangle, \langle B, (S, L) \rangle, \langle C, (L, S) \rangle, \langle D, (L, S) \rangle, \langle E, (L, S) \rangle \}.
        \]

        In this case, agent $B$ would be \emph{satisfied} with a large jacket but would \emph{prefer} a small one, while agent $C$ prefers a large jacket but would accept a small one.
        A fairness measure can require that every agent receives at least one resource from their preference list:

        \[
            \fFair (O) = [\forall a \in \sAgent , \exists b \in \sResource \text{ such that } a \text{ receives } b].
        \]

        The agents may be satisfied with the outcome

        \[
            O = \{ \langle A, S \rangle, \langle B, L \rangle, \langle C, S \rangle, \langle D, L \rangle, \langle E, L \rangle \},
        \]

        but this does not prevent agent $C$ from envying $B$'s jacket, or $B$ from envying $C$'s jacket.
        If they exchange jackets, yielding

        \[
            O = \{ \langle A, S \rangle, \langle B, S \rangle, \langle C, L \rangle, \langle D, L \rangle, \langle E, L \rangle \},
        \]

        i.e., no agent envies another agent's jacket.
        The concept in which no agent envies the outcome of another is called \emph{envy-freeness}~\cite{Foley-1967-ResourceAllocation,Amanatidis-2018-ComparingEnvyFreeness,Richter.Rubinstein-2020-EconomicTheory,Li-2024-PriceEnvyFreeness}.
        We can define a new fairness measure $\fFair (O)$ considering what we call \emph{weak envy-freeness}, where each agent receives at least one resource that they prefer over any resource received by another agent:

        \begin{itemize}
            \renewcommand\labelitemi{}
            \item $\sAgent = \{A, B, C, D, E\}$,
            $\sResource = \{ L, S \}$,
            \item $\sAgentAttribute = \{ v : \sAgent \to \fPerm(\sResource), v = \{ \langle A, (S, L) \rangle,$ $ \langle B, (S, L) \rangle,$ $ \langle C, (L, S) \rangle,$ $ \langle D, (L, S) \rangle,$ $ \langle E, (L, S) \rangle \} \}$,
            \item $\sResourceAttribute = \emptyset$,
            \item $\fFair (O) = [ $ $\nexists a,$ $ a' \in \sAgent, $ $b' \in \sResource $ such that $ a' \text{ receives } b',$ $a \text{ does not receive } b',$ and $\forall b \in \sResource \text{ with } a \text{ receiving } b,$ $\ b' \succ_{a} b ]$.
        \end{itemize}

        Here, $v$ represents the resource preferences of each agent, and no resource attribute is required.
    \end{example}

    The change in outcomes in Example~\ref{ex:preferences} constitutes a \emph{Pareto improvement}, since at least one agent is better off without making anyone else worse off.
    An outcome is \emph{Pareto optimal} when no further Pareto improvement can be applied.

    \subsection{Group and Individual Fairness}
    \label{subsec:group-and-individual}

    \emph{Group fairness} ensures that different demographic groups, which may have protected attributes such as race or gender, receive similar outcomes.
    It focuses on statistical parity across groups.
    \emph{Individual fairness} ensures that similar individuals receive similar outcomes, emphasizing consistency in treatment based on relevant features regardless of group membership.
    Although group fairness and individual fairness may appear to conflict~\cite{Binns-2020-ConflictIndividualGroupFairness}, they can be understood as complementary.

    Let us assume that an attribute can be considered either \emph{relevant} or \emph{irrelevant} in determining the distribution of a resource.
    According to group fairness, if an attribute $p$ is irrelevant, the groups with attribute $p$ should receive the same amount as the groups without attribute $p$.
    According to individual fairness, if an attribute $q$ is relevant, agents with similar values of $q$ should be treated similarly.

    We illustrate this with the following example.

    \begin{example}
        \name{Group and Individual Fairness}
        Assume that a group of agents $\sAgent = \{A, B, C, D, E, F \}$ apply for a loan ($L$), and that there is a protected demographic attribute $p$, which should be irrelevant for the loan application.
        Suppose that only $D$, $E$, and $F$ have this attribute, creating two demographic groups: $G_{\lnot p} = \{ A, B, C\}$ and $G_{p} = \{D, E, F \}$.

        Under group fairness, agents belonging to different demographic groups should be treated similarly, regardless of group membership.
        Assuming that 2 out of 3 loan applications are accepted, this proportion should hold in both $G_{\lnot p}$ and $G_{p}$.
        For example, if the applications of $A$, $B$, $E$, and $F$ are accepted while $C$ and $D$ are rejected, group fairness is satisfied.

        Under individual fairness, if two applicants have nearly identical values for critically relevant attributes, they should receive similar treatment.
        Suppose $q$ is an essential attribute for determining loan eligibility, and that only $B$, $C$, $E$, and $F$ possess this attribute.
        If $D$ receives the loan and $F$ does not, individual fairness is violated in this outcome.

        We formalize this example as follows:

        \begin{itemize}
            \renewcommand\labelitemi{}
            \item $\sAgent = \{A, B, C, D, E, F \}$,
            $\sResource = \{ L \}$,
            \item $\sAgentAttribute = \{ p : \sAgent \to \sBoolean$, $p = \{ $ $\langle A , false \rangle, $ $\langle B , false \rangle, $ $\langle C , false \rangle, $ $\langle D , true \rangle, $ $\langle E , true \rangle, $ $\langle F , true \rangle \}$, $q : \sAgent \to \sBoolean$, \\ $q = \{ \langle A , false \rangle, $ $\langle B , true \rangle, $ $\langle C , true \rangle, $ $\langle D , false \rangle, $ $\langle E , true \rangle, $ $\langle F , true \rangle \} \}$,
            \item $\sResourceAttribute = \emptyset$, \\
            $\varepsilon = 10^{-2}, \ \cdot \simeq_{\varepsilon} \cdot : \sMeasure \times \sMeasure \to \sBoolean$, \\
            $a \simeq_{\varepsilon} b =
            \begin{cases}
                true, & \text{if } a = b \text{ or } (a \neq b \text{ and } \displaystyle \frac{|a - b|}{\max(|a|, |b|)} < \varepsilon) \\
                false, & \text{otherwise}
            \end{cases}$
            \item $P_{\sAgent} = \{ a \in \sAgent \mid p(a) \}$, $\bar P_{\sAgent} = \{ a \in \sAgent \mid \lnot p(a) \}$, $L_{\sAgent} = \{ a \in \sAgent \mid a \text{ receives } L \}$,
            \item $G \fFair (O) = \left[\displaystyle \frac{| P_{\sAgent} \cap L_{\sAgent} |}{ | P_{\sAgent} | } \simeq_{\varepsilon} \frac{| \bar P_{\sAgent} \cap L_{\sAgent} |}{ | \bar P_{\sAgent} | } \right]$,
            \item $I \fFair (O) = [ \forall a, a' \in \sAgent, a \neq a' \text{ and } q(a) = q(a') \Rightarrow (a \text{ receives } L \wedge a' \text{ receives } L) \ \lor \ (a \text{ does not receive } L \wedge a' \text{ does not receive } L) ]$.
        \end{itemize}

        Here, $p$ is an irrelevant protected attribute for group fairness, $q$ is an essential attribute for individual fairness, and $\cdot \simeq_{\varepsilon} \cdot$ determines whether two quantities are similar up to a tolerance $\varepsilon$.

        We thus provide two distinct fairness measures: $G \fFair (O)$ for group fairness and $I \fFair (O)$ for individual fairness.
        Group fairness requires parity in acceptance rates across demographic groups, while individual fairness requires that agents with the same relevant attribute value be treated consistently.
    \end{example}

    \subsection{Relating Envy-Freeness, Individual Fairness, and Group Fairness}

    Group fairness and individual fairness are central concepts in discussions of fairness.
    We now connect these notions to the concept of envy-freeness.
    The generalized group fairness measure is based on the similarity of average received utility between agents with and without a protected attribute.

    \begin{definition}
        \normalfont
        \name{Generalized Group Fairness}
        \label{def:group-fairness}
        Let $p$ be a predicate indicating whether an agent has a protected attribute.
        Define
        $P_{\sAgent} = \{ a \in \sAgent \mid p(a) \}$ and
        $\bar P_{\sAgent} = \{ a \in \sAgent \mid \lnot p(a) \}$,
        where $\sAgent$ is the finite set of agents and $\sResource$ the finite set of resources.
        Let $u : \sAgent \times \sResource \to \mathbb{R}$ be a utility function representing agents' preferences over resources.
        Generalized group fairness is satisfied iff the following holds:
        \begin{equation*}
            \label{eq:generalized-group-fairness}
            \frac{\displaystyle \sum_{a \in P_{\sAgent}} \sum_{b \in \sResource} u(a, b)}{|P_{\sAgent}|}
            \;\simeq_{\varepsilon}\;
            \frac{\displaystyle \sum_{a \in \bar P_{\sAgent}} \sum_{b \in \sResource} u(a, b)}{|\bar P_{\sAgent}|}.
        \end{equation*}

        Here, $\simeq_{\varepsilon}$ denotes approximate equality within tolerance $\varepsilon$.

        Note that if for all $a \in \sAgent, b \in \sResource$, it holds that
        $u(a,b) = [a \text{ receives } b]$,
        then generalized group fairness reduces to the standard group fairness definition.
    \end{definition}

    Some properties of group fairness are outlined below.

    \begin{proposition}
        \normalfont
        \name{Generalized Group Fairness Preserves Envy-Freeness}
        \label{prop:group-fairness-envy-freeness}
        Group fairness ensures envy-freeness among agents that differ only in their protected attribute.
    \end{proposition}

    \begin{proof}
        Consider Definition~\ref{def:group-fairness}.
        Envy-freeness requires that every agent receives at least one resource such that, for that agent, it is no worse than any resource received by another agent.
        Preferences are given by the score $u : \sAgent \times \sResource \to [0, 1]$, where 1 denotes the most preferred and 0 the least preferred.

        Formally, if agent $a$ does not receive a resource $r'$ that another agent $a'$ has received, then there must exist some resource $r$ received by $a$ such that $u(a,r') \leq u(a,r)$.

        Now consider agents $a$ and $a'$ that differ only in their protected attribute.
        Assume envy-freeness does not hold, i.e., $a'$ receives $r'$ but $a$ does not, and $r'$ is strictly more preferred for $a$ than any resource $a$ has received.
        This contradicts group fairness, which ensures that $a$ should have received $r'$ or an equivalent resource.
        Hence, the assumption that envy-freeness does not hold leads to a contradiction.
    \end{proof}

    Envy-freeness is not only related to group fairness but also to individual fairness, as shown below.

    \begin{proposition}
        \normalfont
        \name{Individual Fairness Implies Envy-Freeness}
        \label{prop:individual-fairness-envy-freeness}
        Individual fairness ensures envy-freeness among agents that belong to the same category when classified by a relevant attribute.
    \end{proposition}

    \begin{proof}
        Consider agents $a$ and $a'$ that share identical values for all relevant attributes.
        Assume envy-freeness does not hold, i.e., $a'$ receives $r'$ while $a$ does not.
        This violates individual fairness, since $a$ and $a'$ coincide in their relevant attributes and should therefore be treated consistently.
        The contradiction arises from the assumption that envy-freeness does not hold.
    \end{proof}

    \subsection{Continuous Fairness Measures}
    \label{subsec:continuous-fairness-measures}

    In Definition~\ref{def:fairness-measure}, we defined $\fFair$ as ranging over the continuous interval $\sZeroOneInterval$, but so far we have only presented discrete examples.
    Continuous measures can also be modeled.
    One such example is Jain's fairness index~\cite{Jain-1984-QuantitativeMeasureOfFairness}, a quantitative measure used to assess how evenly a resource is allocated among $n$ agents.

    \begin{definition}
        \normalfont
        \name{Jain's Fairness Index}
        \label{def:jain-index}
        Given a set of agents indexed from 1 to $n$, such that each agent receives $x_{1}, x_{2}, \ldots , x_{n}$ respectively, the index is defined (left) and rewritten in terms of fairness measures (right) as:
        \begin{equation}
            \label{eq:jain-equation}
            \begin{tabular}{cc}
                $ J(x_{1}, \ldots , x_{n}) = \displaystyle \frac{\bigl(\sum \limits _{i=1}^{n} x_{i} \bigr)^2}{n \cdot \sum \limits_{i=1}^{n} x_{i}^{2}} \ , $
                &
                $ \fFair(O) = \displaystyle \frac{\bigl(\sum \limits _{a \in \sAgent} r_{O}(a) \bigr)^2}{|\sAgent| \cdot \sum \limits_{a \in \sAgent} r_{O}(a)^{2}} \ $.
                \\
            \end{tabular}
        \end{equation}
    \end{definition}

    Jain's fairness index is widely used to measure fairness in network resource allocation, evaluate load balancing schemes in distributed systems, and balance throughput in congestion control protocols.

    \begin{example}
        \name{Continuous Fairness Measure}
        Assume that agents need to access resources that represent different values of bandwidth in a computer network.

        \begin{itemize}
            \renewcommand\labelitemi{}
            \item $\sAgent = \{A, B, C, D \}$,
            $\sResource = \{ M_{0}, M_{10}, M_{20}, M_{50} \}$,
            \item $\sAgentAttribute = \emptyset$,
            \item $\sResourceAttribute = \{ u : \sResource \to \sMeasure, $ $\ u = \{ \langle M_{0}, 0 \rangle, $ $\langle M_{10}, 10 \rangle, $ $\langle M_{20}, 20 \rangle, $ $\langle M_{50}, 50 \rangle \} \}$,
            \item $\fFair(O)$ as in (\ref{eq:jain-equation})
        \end{itemize}

        Here, $u$ represents the utility in megabits per second (Mbps) of each resource.

        Considering the following outcomes:

        \[
            O_{1} = \{ \langle A, M_{20} \rangle, \langle B, M_{20} \rangle, \langle C, M_{20} \rangle, \langle D, M_{20} \rangle \},
        \]

        \[
            O_{2} = \{ \langle A, M_{20} \rangle, \langle B, M_{20} \rangle, \langle C, M_{20} \rangle, \langle D, M_{0} \rangle \},
        \]

        \[
            O_{3} = \{ \langle A, M_{0} \rangle, \langle B, M_{0} \rangle, \langle C, M_{0} \rangle, \langle D, M_{10} \rangle \},
        \]

        we obtain $\fFair(O_{1}) = 1$, $\fFair(O_{2}) = 0.75$, and $\fFair(O_{3}) = 0.25$.
        Intuitively, $O_{1}$ is perfectly fair, $O_{2}$ moderately fair, and $O_{3}$ clearly unfair.
        In practice, however, the interpretation of these values depends strongly on context.
    \end{example}

    A continuous fairness measure more closely aligned with \emph{social} (as opposed to \emph{technical}) applications is the Gini index, which we discuss in the following example.
    \begin{example}
        \name{Complement of the Gini Index}
        \label{ex:gini-index}
        The \emph{Gini index} is a statistical measure of inequality commonly used to quantify the distribution of income or wealth within a population.
        It has also been applied in other contexts, such as decision tree algorithms in machine learning.
        In its interpretation, values closer to 0 indicate equality, while values closer to 1 indicate inequality.
        This behavior is opposite to Jain's index and to our definition of a fairness measure.
        Therefore, we define the \emph{complement of the Gini index} by inverting its output.
        The Gini index is defined (left) and rewritten in terms of fairness measures (right) as:

        \begin{equation}
            \label{def:gini-equation}
            \resizebox{\textwidth}{!}{
                \begin{tabular}{cc}
                    $G = \displaystyle\frac{\sum_{i=1}^{n} \sum_{j=1}^{n} |x_{i} - x_{j}|}{2 \cdot n \cdot \sum_{i=1}^{n} x_{i}} \ , $
                    &
                    $\fFair (O) = 1 - \displaystyle\frac{\sum_{a_{1} \in \sAgent} \sum_{a_{2} \in \sAgent} |r_{O}(a_{1}) - r_{O}(a_{2})|}{2 \cdot |\sAgent| \cdot \sum_{a \in \sAgent} r_{O}(a)}$ \\
                \end{tabular}
            }
        \end{equation}

        Here, $n$ is the number of agents (or households), and $x_{i}$ is the income (or wealth) of agent $i$.

        As in Jain's index (Definition~\ref{def:jain-index}), we substitute $\sAgent$ for $n$ and $r_{O}(a)$ for the amount received by each agent $a$.
        In this context, the agents represent households, and the resource is wealth to be distributed.

        \begin{itemize}
            \renewcommand\labelitemi{}
            \item $\sAgent = \{A, B, C, D, E, F \}$,
            $\sResource = \{ R_{5}, R_{10}, R_{15}, R_{20}, R_{50}, R_{100} \}$,
            \item $\sAgentAttribute = \emptyset$,
            \item $\sResourceAttribute = \{ u : \sResource \to \sMeasure,$ $ \ u = \{ $ $\langle R_{5}, 5 \rangle, $ $\langle R_{10}, 10 \rangle, $ $\langle R_{15}, 15 \rangle, $ $\langle R_{20}, 20 \rangle, $ $\langle R_{50}, 50 \rangle, $ $\langle R_{100}, 100 \rangle $ $\} \}$,
            \item $\fFair(O)$ as in (\ref{def:gini-equation})
        \end{itemize}

        Here, $u$ represents the utility (e.g., in thousands of euros) of each resource.

        Considering the following outcomes:

        \[
            O_{1} = \{ \langle A, R_{20} \rangle, \langle B, R_{20} \rangle, \langle C, R_{20} \rangle, \langle D, R_{20} \rangle, \langle E, R_{20} \rangle, \langle F, R_{20} \rangle \},
        \]

        \[
            O_{2} = \{ \langle A, R_{5} \rangle, \langle B, R_{10} \rangle, \langle C, R_{15} \rangle, \langle D, R_{20} \rangle, \langle E, R_{50} \rangle, \langle F, R_{100} \rangle \},
        \]

        \[
            O_{3} = \{ \langle A, R_{5} \rangle, \langle B, R_{10} \rangle \},
        \]

        we obtain $\fFair(O_{1}) = 1 - \frac{0}{1440} = 1$, $\fFair(O_{2}) = 1 - \frac{1200}{2400} = 0.5$, and $\fFair(O_{3}) = 1 - \frac{130}{180} \approx 0.28$.
        Intuitively, $O_{1}$ is perfectly fair, $O_{2}$ moderately unfair, and $O_{3}$ clearly unfair.
    \end{example}

    Other indices measure economic inequality in addition to the Gini index.
    Many of them quantify deviations from the average.
    One such measure is the Theil index.

    \begin{example}
        \name{Complement of the Theil Index}
        The \emph{Theil index}~\cite{Conceicao.Ferreira-2000-Theil} is a statistical measure of economic inequality based on the concept of entropy.
        As with the Gini index, values closer to 0 indicate equality, while values closer to 1 indicate inequality.
        For consistency with our fairness measure definition, we define the \emph{complement of the Theil index} by inverting its output.
        The Theil index is defined (left) and rewritten in terms of fairness measures (right) as:

        \begin{equation}
            \label{def:theil-equation}
            \resizebox{\textwidth}{!}{
                \begin{tabular}{cc}
                    $T_{T} = \displaystyle
                    \frac{1}{n} \cdot \sum_{i=1}^{n} \frac{x_{i}}{\overline{x}} \cdot \ln \left(\frac{x_{i}}{\overline{x}} \right) \ , $
                    &
                    $\fFair (O) = 1 - \displaystyle \frac{1}{|\sAgent|} \cdot \sum_{a \in \sAgent} \frac{r_{O}(a)}{\overline{x}} \cdot \ln \left(\frac{r_{O}(a)}{\overline{x}} \right)$ \\
                \end{tabular}
            }
        \end{equation}

        where $\overline{x}$ is the average value, defined as:

        \begin{equation}
            \label{def:theil-equation-average}
            \begin{tabular}{cc}
                $\overline{x} = \displaystyle \frac{1}{n} \cdot \sum_{i=1}^{n} x_{i} \ , $
                &
                $\overline{x} = \displaystyle \frac{1}{|\sAgent|} \cdot \sum_{a \in \sAgent} r_{O}(a)$ \\
            \end{tabular}
        \end{equation}

        Here, $n$ is the number of agents (or households) and $x_{i}$ is the income (or wealth) of agent $i$.

        \begin{itemize}
            \renewcommand\labelitemi{}
            \item $\sAgent = \{A, B, C, D, E, F \}$,
            $\sResource = \{ R_{5}, R_{10}, R_{15}, R_{20}, R_{50}, R_{100} \}$,
            \item $\sAgentAttribute = \emptyset$,
            \item $\sResourceAttribute = \{ u : \sResource \to \sMeasure, $ $ u = \{ \langle R_{5}, 5 \rangle,$ $ \langle R_{10}, 10 \rangle,$ $ \langle R_{15}, 15 \rangle,$ $ \langle R_{20}, 20 \rangle,$ $ \langle R_{50}, 50 \rangle,$ $ \langle R_{100}, 100 \rangle \} \}$,
            \item $\fFair(O)$ as in (\ref{def:theil-equation})
        \end{itemize}

        For comparison, consider the first two outcomes from Example~\ref{ex:gini-index}:
        \begin{enumerate}
            \item $\overline{x} = 120/6 = 20$, yielding $\fFair(O_{1}) = 1 - 0 = 1$,
            \item $\overline{x} = 200/6 \simeq 33.33$, yielding $\fFair(O_{2}) = 1 - 0.432 \approx 0.568$.
        \end{enumerate}

        While the first outcome is measured as perfectly fair, consistent with the Gini index, the second outcome is assessed as somewhat fairer than under the Gini index.
        This illustrates how different continuous fairness measures can yield distinct interpretations of the same distribution.
    \end{example}

    Continuous inequality measures can be unified under a generalized formulation known as the generalized inequality index (GII)~\cite{Tsekouras_Tsallis-2005-GeneralizedEntropy,Dagum-1997-GeneralizedEntropyInequalityMeasures,MussardSeyteTerraza-2003-GeneralizedEntropy}.

    \begin{definition}
        \normalfont
        \name{Generalized Inequality Index}
        \label{def:gen-ineq-index}
        A \emph{generalized inequality index} (GII) is a function ranging over $[0,1]$.
        On the left we present the definition of a GII, and on the right its reformulation in terms of fairness measures:
        \begin{equation}
            \label{def:general-inequality}
            \begin{tabular}{cc}
                $ \mathrm{GII}= \displaystyle \sum_{i} p_{i} \cdot f\!\left( \frac{x_{i}}{\overline{x}} \right)$ \ , &
                $\fFair (O) = 1 - \displaystyle \sum_{a \in \sAgent} w(a) \cdot f\!\left(\frac{r_{O}(a)}{\overline{x}}\right)$ \\
            \end{tabular}
        \end{equation}

        Here, $x_{i}$ denotes the amount received by agent $i$, $\overline{x}$ is the mean across agents as defined in (\ref{def:theil-equation-average}), $p_{i}$ denotes the weight associated with agent $i$, $w(a)$ denotes the weight for agent $a$, and $f(x)$ is chosen according to the specific inequality measure.
        For example, in the Theil index defined in (\ref{def:theil-equation}), $p_{i} = w(a) = \tfrac{1}{n}$ and $f(x) = x \cdot \ln(x)$.

        Since the GII measures inequality, a value of 0 corresponds to perfect equality, while positive values increase with inequality.
        In the \AR framework, we invert these values: perfect equality is mapped to 1, and lower non-negative values indicate decreasing fairness.
    \end{definition}

    Continuous fairness measures can also be combined to fine-tune evaluations and capture multiple contexts of fairness simultaneously.
    The \emph{contextual fairness}~\cite{Kerkhoven-2025-Assessing} measure is a fairness measure composed by other measures.

    \begin{definition}
        \normalfont
        \name{Contextual Fairness}
        \label{def:contextual-fairness}
        Given a set of fairness measures $\fFair_{i}$ with respective weights $w_{i}$, $1 \leq i \leq n$, such that
        \[
            \sum_{i=1}^{n} w_{i} = 1.
        \]

        We denote the \emph{contextual fairness measure} $\fFair_{\mathrm{co}}$ as:
        \[
            \fFair_{\mathrm{co}}(O) = \sum_{i=1}^{n} w_{i}\cdot \fFair_{i}(O).
        \]
    \end{definition}

    Fairness measures as defined in Definition~\ref{def:contextual-fairness} can be used to detect bias in a dataset.
    In Example~\ref{ex:compas}, we show a salient case where bias was found in a computer system used for justice in the United States.

    \begin{example}
        \name{Fairness Measure for Equalized Odds}
        \label{ex:compas}
        The COMPAS (Correctional Offender Management Profiling for Alternative Sanctions) system~\cite{ProPublica-2016} is a well-known risk assessment tool used in the U.S. criminal justice system to predict the likelihood of recidivism (re-offending after a prior arrest).
        Studies, most notably by ProPublica in 2016\footnote{\url{https://github.com/propublica/compas-analysis}}, found that COMPAS exhibited racial bias: it was more likely to falsely flag Black defendants as future criminals (higher false positive rate) and more likely to misclassify White defendants as low risk (higher false negative rate).

        Based on this example, we design a fairness measure to detect bias in a prediction system.
        The system outputs two possible scores (resources), analogous to COMPAS:
        \begin{itemize}
            \item $R_{\mathrm{low}}$: Low Risk (scores 1--4)
            \item $R_{\mathrm{high}}$: Medium Risk (scores 5--7) and High Risk (scores 8--10)
        \end{itemize}

        The prediction system is defined as follows:
        \begin{itemize}
            \renewcommand\labelitemi{}
            \item $\sAgent = \{A, B, C, D, E, F \}$,
            $\sResource = \{ R_{\mathrm{low}}, R_{\mathrm{high}} \}$,
            \item $\sAgentAttribute = \{ p : \sAgent \to \sBoolean, $ $ p = \{ \langle A , false \rangle,$ $ \langle B , false \rangle,$ $ \langle C , false \rangle,$ $ \langle D , true \rangle,$ $ \langle E , true \rangle,$ $ \langle F , true \rangle \},$
            \item $res : \sAgent \to \sResource, $ $ res = \{ \langle A , R_{\mathrm{low}} \rangle,$ $ \langle B , R_{\mathrm{high}} \rangle,$ $ \langle C , R_{\mathrm{high}} \rangle,$ $ \langle D , R_{\mathrm{low}} \rangle,$ $ \langle E , R_{\mathrm{low}} \rangle,$ $ \langle F , R_{\mathrm{high}} \rangle \},$
            \item $\sResourceAttribute = \{ u : \sResource \to \sMeasure, \ u = \{ \langle R_{\mathrm{low}}, 0 \rangle,$ $ \langle R_{\mathrm{high}}, 1 \rangle \} \}$,
            \item $\overline{(x_{i})_{i=1}^{n}}$ denotes the arithmetic mean of the sequence $(x_{1}, x_{2}, \ldots , x_{n})$,
            \item $corr : \bigcup_{n=1}^{\infty} (\sMeasure \times \sMeasure) \to [-1,1]$, defined as \\
            $corr((x_{i})_{i=1}^{n}, (y_{i})_{i=1}^{n}) =
            \frac{\sum_{k=1}^{n} (x_{k} - \overline{(x_{i})_{i=1}^{n}})(y_{k} - \overline{(y_{i})_{i=1}^{n}})}
            {\sqrt{\sum_{k=1}^{n} (x_{k} - \overline{(x_{i})_{i=1}^{n}})^2} \cdot
            \sqrt{\sum_{k=1}^{n} (y_{k} - \overline{(y_{i})_{i=1}^{n}})^2}}$

            \item $\fFair(O) = | corr(([\langle x_{i}, R_{\mathrm{high}} \rangle \in O \land res(x_{i}) \neq R_{\mathrm{high}}])_{i=1}^{n}, ([p(x_{i})])_{i=1}^{n}) |$
        \end{itemize}

        Here, $p$ is the protected attribute (e.g., race in COMPAS), $res$ is the ground truth based on facts, and $u$ is the associated risk of each score.
        We use the Pearson correlation coefficient, though other correlation coefficients could also be applied.
        Recall that the square brackets $[\,]$ follow Iverson notation, and the fairness measure is constrained to return values in $[0,1]$.
    \end{example}

    \section{Operationalization}
    \label{sec:operationalization}

    In this section, we explain how the concepts presented above are operationalized, that is, how abstract ideas are translated into concrete, formal representations that can be implemented and measured.
    An instance of the \AR metamodel provides a formal representation of a fairness scenario
    $\tFairScen = \langle \sAgent, \sResource, \sAgentAttribute, \sResourceAttribute \rangle$.
    Assuming that the agent and resource attributes are given, the fairness measure still needs to be defined.
    As illustrated in the previous examples, defining a fairness measure can be complex and error-prone.

    To address this challenge, we employ the \Tiles framework~\cite{Mendez.Kampik.Aler.Dignum-2024-SCAI} as a means of specifying fairness measures for a given fairness scenario.
    The purpose of this framework is to improve both the readability and the reliability of fairness measures.
    To achieve this, we decompose the fairness measure into smaller, composable building blocks.

    \subsection{Design of the Blocks}
    \label{subsec:design}

    The design of blocks in \Tiles is inspired by the concept of a \emph{module} in software engineering.

    \begin{concept}
        \normalfont
        \name{Module}
        A \emph{module} is a software component characterized by the following properties:
        \begin{enumerate}
            \item (\emph{encapsulation}) hides internal details, exposing only necessary interfaces;
            \item (\emph{reusability}) can be applied across different systems or projects to reduce redundancy;
            \item (\emph{interchangeability}) can be replaced by other modules with equivalent functionality;
            \item (\emph{cohesion}) is focused on a single, well-defined purpose;
            \item (\emph{low coupling}) minimizes dependencies on unrelated components;
            \item (\emph{scalability}) adapts to future changes with minimal impact;
            \item (\emph{testability}) is designed for independent verification of functionality.
        \end{enumerate}
    \end{concept}

    Recall that a fairness scenario is a tuple $\tFairScen = \langle \sAgent, \sResource, \sAgentAttribute, \sResourceAttribute \rangle$.
    We assume that $\sResourceAttribute$ contains a function $u : \sResource \to \sMeasure$ \emph{(utility)}, $\sAgentAttribute$ contains a function $q : \sAgent \to \sMeasure$ \emph{(needs)}, and we define the auxiliary function $r_{O} : \sAgent \to \sMeasure$ \emph{(accumulates)} as in Definition~\ref{def:equality-equity}, which depends on $u$.

    The functional notation in \Tiles can be represented using a graphical notation, making the pipeline structure visually explicit.

    \subsection{Graphical Notation}
    \label{subsec:notation}

    One of the key aspects of \Tiles is its ability to represent configurations through a graphical notation that clarifies how a configuration operates.
    A configuration is a formal representation of how blocks are interconnected.
    Each block, called a \emph{tile}, can connect with others to define a specific construct.

    \begin{concept}
        \normalfont
        \name{Tile}
        \label{def:tile}
        A \emph{tile} is a construct that contains a name, a function, an input type, an output type, and contextual information.
        The contextual information includes the fairness scenario, constants, and auxiliary functions.
        A tile is represented as
        \[
            \tiles{_{\alpha} \ name \ _{\beta}}
        \]
        where $\alpha$ and $\beta$ are type annotations for the input and output types, respectively, and $\tilesfun{name}$ denotes the function name.
        The input type is omitted if the tile represents a constant. \\

        \noindent A type in \Tiles can be:
        \begin{itemize}
            \item atomic;
            \item a tuple composed of other types: $\langle \alpha_{1}, \ldots, \alpha_{n} \rangle$; or
            \item a sequence of a type: $(\alpha)$.
        \end{itemize}
        The atomic types for fairness scenarios are: \tilestype{a} (agent), \tilestype{r} (resource), \tilestype{m} (quantity or measure), and \tilestype{b} (Boolean).
    \end{concept}

    Type annotations in \Tiles can be used not only to specify the connection between the output of one tile and the input of another, but also to denote input variable names when the tile has parameters.
    Parametric tiles are particularly useful for defining customized tiles.

    The \Tiles framework provides primitive tiles, as shown in Table~\ref{tab:primitive-tiles}. When types coincide or can be unified, the tiles can be snapped together. In that case, the most specific type that unifies in denoted in the connection.

    If some parameters are required by an internal tile and remain constant by the previous ones, they can be denoted altogether following a semicolon (;) after the necessary parameters.

    \begin{longtable}{ll}
        \caption{Definitions and examples of primitive tiles} \\
        \hline
        \textbf{Function} & \textbf{Description}                                                                      \\
        \hline
        &                                                                                           \\
        $\text{cross}((a_{i})_{i=1}^{n},(b_{j})_{j=1}^{m})$
        & \makecell[l]{$(\langle a_{1}, b_{1}\rangle,$ $ \langle a_{1}, b_{2}\rangle,$ $ \ldots , $ \\ $ \langle a_{1}, b_{m} \rangle,$ $ \langle a_{2}, b_{1} \rangle,$ $ \ldots,$ \\ $ \langle a_{n}, b_{m}\rangle)$}                  \\
        $\text{cross}((A, B, C),(1, 2, 3))$                  &
        \makecell[l]{$(\langle A, 1 \rangle ,$ $ \langle A, 2 \rangle ,$  $ \langle A, 3 \rangle ,$ \\ $ \langle B , 1 \rangle,$ $ \langle B, 2 \rangle ,$ $ \langle B, 3 \rangle ,$ \\ $ \langle C , 1 \rangle,$ $ \langle C , 2 \rangle, $ $ \langle C , 3 \rangle)$}                  \\
        \\
        $\text{zip}((a_{i})_{i=1}^{n},(b_{j})_{j=1}^{m})$
        & $(\langle a_{i},b_{i}\rangle)_{i=1}^{\min \{n,m\}}$                                  \\
        $\text{zip}((A, B, C, D), (1,2,3))$
        & $(\langle A, 1\rangle, \langle B, 2\rangle , \langle C, 3 \rangle)$                                  \\
        \\
        $\text{apply}_{\varphi}(a)$
        & $\varphi(a)$                                                                                                                                                                       \\
        $\text{apply}_{x \mapsto x > 0}(0)$
        & $false$                                                                                                                                                                       \\
        & \\
        $\text{bind}_{\varphi}((a_{i})_{i=1}^{n})$
        & $\varphi(a_{1}) \ || \ \varphi(a_{2}) \ || \ \ldots \ || \ \varphi(a_{n})$                                                                                                                                                                       \\
        $\text{bind}_{x \mapsto (x, x+1)}((0, 1, 7))$
        & $(0, 1, 1, 2, 7, 8)$                                                                                                                                                                       \\
        & \\
        $\text{fold}_{\varphi}(z,(a_{i})_{i=1}^{n})$
        & $\varphi(\ldots\varphi(\varphi(z,a_{1}),a_{2})\ldots,a_{n})$                                                                                                                                                                       \\
        $\text{fold}_{\langle z, x \rangle \mapsto z + x}(0, (1,2,4,8,16))$
        & $31 $                                                                                                                                                                       \\
        & \\
        & \\
        where $\cdot || \cdot$ is the \\
        concatenation of sequences. \\
        & \\
        \hline
        \label{tab:primitive-tiles}
    \end{longtable}

    Basic auxiliary tiles are also provided, as defined in Table~\ref{tab:basic-auxiliary-tiles}.

    \begin{longtable}{ll}
        \caption{Definitions and examples of basic auxiliary tiles} \\
        \hline
        \textbf{Function} & \textbf{Description}                                                         \\
        \hline
        &                                                                              \\
        $\text{filter}_{\varphi}((a_{i})_{i=1}^{n})$
        & \makecell[l]{$\text{bind}_{f_\varphi} ((a_{i})_{i=1}^{n})$, \\
            $f_\varphi(x) = \begin{cases}
                                (x), & \text{ if } \varphi(x) \\ () , & \text { otherwise }
            \end{cases}$
        } \\
        $\text{filter}_{x \mapsto \textrm{even}(x)}((1, 4, 1, 4, 2))$
        & \makecell[l]{$(4, 4, 2)$}                                                    \\
        \\
        $\text{map}_{\varphi}((a_{i})_{i=1}^{n})$
        & $\text{bind}_{f_\varphi} ((a_{i})_{i=1}^{n})$, $f_\varphi(x) = (\varphi(x))$ \\
        $\text{map}_{x \mapsto 2 x}((0, 1, 1, 2, 3))$
        & $(0, 2, 2, 4, 6)$                  \\
        & \\
        $\text{reverse}((a_{i})_{i=1}^{n})$
        & $\text{fold}_{f}((a_{i})_{i=1}^{n})$, $f(z, x) = (x) \ || \ z$                                                                                  \\
        $\text{reverse}((0, 1, 1, 2, 3))$
        & $(3, 2, 1, 1, 0)$                  \\
        & \\
        $\text{distinct}((a_{i})_{i=1}^{n})$
        & \makecell[l]{$\text{fold}_{f}((a_{i})_{i=1}^{n})$, \\
            $f(z, x) = \begin{cases}
                           z, & \text{ if } x \in z \\ z \ || \ (x) , & \text { otherwise }
            \end{cases}$
        } \\
        $\text{distinct}((A, B, B, A, C))$
        & $(A, B, C)$                  \\
        & \\
        \hline
        \label{tab:basic-auxiliary-tiles}
    \end{longtable}

    \subsection{Examples}

    Let us now consider some examples of how \Tiles operates.
    Figure~\ref{fig:pipeline-equality-primitive} illustrates a pipeline for equality.
    Define $\varphi_{c} : \sMeasure \times \sMeasure \to \sMeasure$ by $\varphi_{c}(z,x) = z + 1$.

    \begin{figure*}[ht!]
        \centering
        \resizebox{\textwidth}{!}{
            \begin{tikzpicture}[x=1mm, y=1mm, box/.style={rectangle, draw, rounded corners=2mm, minimum width=4mm, minimum height=8mm, align=center}]
            \node [box] (S) {};
            \node [box, right=4mm of S] (allagent) {$\mathsf{\ all\text{-}agent \ \Big| \ _{(a)} \ map \ \mathnormal{r_{O}} \ \Big| \ _{(m)} \ distinct \ \Big| \ _{(m)} \ fold \ 0 \ using \ \mathnormal{\varphi_{c}} \ \Big| \ _{m} \ apply \ m=1 \ _{b}} $};
            \draw [->] (S.east) -- (allagent.west);
            \end{tikzpicture}
        }

        \caption{Pipeline for equality.}
        \label{fig:pipeline-equality-primitive}
    \end{figure*}
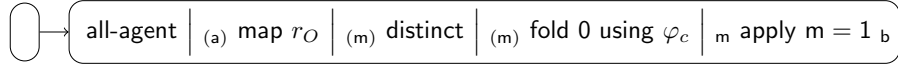

    To improve readability, we can introduce the following derived tiles:

    \begin{equation*}
        \tiles{_{(a)} \ accumulates \ _{(m)}} := \tiles{_{(a)} \ map \ \mathnormal{r_{O}} \ _{(m)}} \ ,
    \end{equation*}

    \begin{equation*}
        \tiles{_{(\alpha)} \ length \ _{m}} := \tiles{_{(\alpha)} \ fold \ 0 \ using \ \mathnormal{\varphi_{c}} \ _{m}} \ ,
    \end{equation*}

    and the following composite tile:

    \begin{equation*}
        \tiles{_{(m)} \ all\text{-}equal \ _{b}} := \tiles{\ _{(m)} \ distinct \ \Big| \ _{(m)} \ length \ \Big| \ _{m} \ apply \ m=1 \ _{b}} \ .
    \end{equation*}

    This yields Figure~\ref{fig:pipeline-equality-short}, which expresses the same pipeline in a more concise and readable form.

    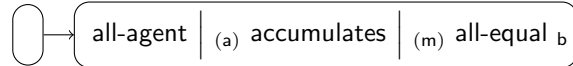
\begin{figure*}[ht!]
        \centering
        \begin{tikzpicture}[x=1mm, y=1mm, box/.style={rectangle, draw, rounded corners=2mm, minimum width=4mm, minimum height=8mm, align=center}]
            \node [box] (S) {};
            \node [box, right=4mm of S] (allagent) {$\mathsf{\ all\text{-}agent \ \Big| \ _{(a)} \ accumulates \ \Big| \ _{(m)} \ all\text{-}equal \ _{b}} $};
            \draw [->] (S.east) -- (allagent.west);
        \end{tikzpicture}
        \caption{Pipeline for equality using derived and composite tiles.}
        \label{fig:pipeline-equality-short}
    \end{figure*}

    Table~\ref{tab:tiles-customized} presents a summary of customized tiles used in the examples.

    \begin{longtable}{ll}
        \caption{Customized tiles based on primitive and basic auxiliary tiles} \\
        \hline
        \textbf{Tile Definition} \\
        \hline
        \\
        \tiles{_{(a)} \ accumulates \ _{(m)}}                                         & :=                            \tiles{_{(a)} \ map \ \mathnormal{r_{O}}(a) \ _{(m)}}          \ \ \ \                                                                                                                                        \\
        \\
        \tiles{_{(a)} \ needs \ _{(m)}}   & :=                                \tiles{_{(a)} \ map \ \mathnormal{q}(a) \ _{(m)}}                                                                        \\
        \\
        \tiles{_{(\alpha )} \ sum \ \varphi (\mathnormal{x_{\alpha}}) \ _{m}}                                    & :=  \tiles{_{(\alpha)} \ fold \ 0 \ using \ \langle \mathnormal{z}, \mathnormal{x_{\alpha}} \rangle \mapsto \mathnormal{z} + \varphi (\mathnormal{x_{\alpha}}) \ _{m}}                   \\
        & \qquad where $\varphi : \alpha \to \sMeasure$, and $x_{\alpha} \in \alpha$                                                                                                                          \\
        \\
        \tiles{_{(m)} \ sum \ _{m_{0}}}   & :=                                   \tiles{_{(m)} \ sum \ m \ _{m_{0}}}                   \\
        \\
        \tiles{_{(\alpha)} \ length \ _{m}} & :=                              \tiles{_{(\alpha)} \ sum \ 1 \ _{m}} \\
        \\
        \tiles{_{(m)} \ all\text{-}equal \ _{(b)}}  & :=                    \tiles{_{(m)} \ distinct \ \Big| \ _{(m)} \ length \ \Big| \ _{m} \ apply \ m \leq 1 \ _{b}}                                 \\
        \\
        \tiles{_{(\alpha)} \ exists \ \varphi (\mathnormal{x_{\alpha}}) \ _{b}} & := \tiles{_{(\alpha)} \ fold \ \mathnormal{false} \ using \ \langle \mathnormal{z}, \mathnormal{x_{\alpha}} \rangle \mapsto \mathnormal{z} \lor \varphi (\mathnormal{x_{\alpha}}) \ _{b}}                  \\
        & \qquad where $\varphi : \alpha \to \sBoolean$, and $x_{\alpha} \in \alpha$                                                                                                                          \\
        \\
        \tiles{_{(\alpha)} \ forall \ \varphi (\mathnormal{x_{\alpha}}) \ _{b}} & := \tiles{_{(\alpha)} \ fold \ \mathnormal{true} \ using \ \langle \mathnormal{z}, \mathnormal{x_{\alpha}} \rangle \mapsto \mathnormal{z} \land \varphi (\mathnormal{x_{\alpha}}) \ _{b}}                 \\
        & \qquad where $\varphi : \alpha \to \sBoolean$, and $x_{\alpha} \in \alpha$                                                                                                                          \\
        \\
        \tiles{_{(m_{0}),(m_{1})} \ all\text{-}at\text{-}least \ _{b}} & := \tiles{_{(m_{0}), (m_{1})} \ zip \ \Big| \ _{(\langle m_{0}, m_{1}\rangle )} \ forall \ (m_{0} \geq m_{1}) \ _{b}}        \\
        \\
        \hline
        \label{tab:tiles-customized}

    \end{longtable}

    It is possible to design components using \Tiles. In that case, a parameter can be given to the component. Figure~\ref{fig:def-rt} contains the definition of a component to compute $r_{O}(\mathsf{a})$ for an agent $\mathsf{a}$ by only using the primitive property
    $\mathsf{a \text{ receives } r}$, for agent $\mathsf{a}$ and a resource $\mathsf{r}$. Notice that, although \tilesfun{all\text{-}resource} does not need $\mathsf{a}$, we use the semicolon notation to denote that the tile passes the parameter on to \tilesfun{filter}, which needs it.

    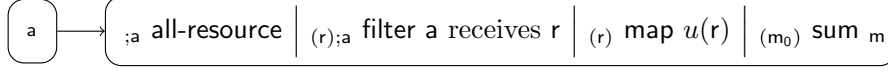
\begin{figure*}[ht!]
        \centering
        \resizebox{\textwidth}{!}{
            \begin{tikzpicture}[x=1mm, y=1mm, box/.style={rectangle, draw, rounded corners=2mm, minimum width=4mm, minimum height=8mm, align=center}]
            \node [box] (S) {$\ \mathsf{_{a}} \ $};
            \node [box, right=6mm of S] (allagent) {$\mathsf{ \ _{;a} \  all\text{-}resource \ \Big| \ _{(r);a} \ filter \ a \text{ receives } r \ \Big| \ _{(r)} \ map \ \mathnormal{u}(r) \ \Big| \ _{(m_{0})} \ sum \ _{m}} $};
            \draw [->] (S.east) -- (allagent.west);
            \end{tikzpicture}
        }

        \caption{Component that computes $r_{O}(\mathsf{a})$, where the agent $\mathsf{a}$ is a parameter for the component.}
        \label{fig:def-rt}
    \end{figure*}

    We represent the agents $\sAgent$ as a sorted sequence of identifiers $\fOrderedSeq{A}$, denoted by \tilesfun{all-agent}.
    We define \tilesfun{accumulates} and \tilesfun{needs} by applying $r_{O}$ and $q$, respectively, to each element of the sequence.
    Additionally, we define \tilesfun{all-equal} as the result of checking whether all elements in the sequence are equal, and \tilesfun{all-at-least} as the result of checking, for each pair in the sequence, whether its first component is greater than or equal to the second.

    Equality can then be expressed through the following pipeline:
    \begin{equation}
        \fFair_{eqa} (O) = \tilesfun{all-equal}(\tilesfun{accumulates}(\tilesfun{all-agent}))
        \label{eq:equality-pipeline}
    \end{equation}

    Thus, $\fFair_{eqa}$ in Definition~\ref{def:equality-equity} is equivalent to $\fFair_{eqa}$ given in (\ref{eq:equality-pipeline}).

    \begin{proposition}
        \normalfont
        \name{Correctness of the Equality Pipeline}
        Let $\tFairScen = \langle \sAgent, \sResource, \sAgentAttribute, \sResourceAttribute \rangle$ be a fairness scenario, $\fFair_{eqa}$ defined as in Definition~\ref{def:equality-equity}, and $\fFair_{eqa'}$ defined as in (\ref{eq:equality-pipeline}).
        Then, for every outcome $O$,

        \[
            \fFair_{eqa} (O) = 1 \Longleftrightarrow \fFair_{eqa'} (O) = 1 \ .
        \]

    \end{proposition}

    \begin{proof}
        Expanding $\fFair_{eqa'} (O)$ in (\ref{eq:equality-pipeline}) yields:
        \begin{equation*}
            [ \forall m, m' \in (r_{O}(a))_{a \in \fOrderedSeq{\sAgent}} \ (m = m') ] \ .
        \end{equation*}

        Since $\sAgent$ is non-empty and finite, $a \in \sAgent$ iff $a \in \fOrderedSeq{\sAgent}$.
        In particular, the sequence $(r_{O}(a))_{a \in \fOrderedSeq{\sAgent}}$ contains, for each agent $a$, the value $r_{O}(a)$.
        Thus, for each value $m$, $m \in (r_{O}(a))_{a \in \fOrderedSeq{\sAgent}}$ iff $m \in \{ r_{O}(a) \mid a \in \sAgent \}$.

        Therefore, $\fFair_{eqa'} (O)$ holds iff

        \[
            \forall m, m' \in \{ r_{O}(a) \mid a \in \sAgent \} \ (m = m') \ .
        \]

        This is equivalent to
        \begin{equation}
            \forall a, a' \in \sAgent \ (r_{O}(a) = r_{O}(a')) \ ,
            \label{eq:equality-of-values}
        \end{equation}
        which is precisely the definition of $\fFair_{eqa} (O)$.
    \end{proof}

    We define equity with the following pipeline:
    \begin{equation}
        \fFair_{eqi} (O) = \tilesfun{all-at-least}(\tilesfun{accumulates}(\tilesfun{all-agent}), \tilesfun{needs}(\tilesfun{all-agent}))
        \label{eq:equity-pipeline}
    \end{equation}

    Observe that $\fFair_{eqi}$ in Definition~\ref{def:equality-equity} is equivalent to $\fFair_{eqi}$ given in (\ref{eq:equity-pipeline}).

    \begin{proposition}
        \normalfont
        \name{Correctness of the Equity Pipeline}
        Let $\tFairScen = \langle \sAgent, \sResource, \sAgentAttribute, \sResourceAttribute \rangle$ be a fairness scenario, $\fFair_{eqi}$ defined as in Definition~\ref{def:equality-equity}, and $\fFair_{eqi'}$ defined as in (\ref{eq:equity-pipeline}).
        Then, for every outcome $O$,
        \[
            \fFair_{eqi} (O) = 1 \Longleftrightarrow \fFair_{eqi'} (O) = 1 \ .
        \]

    \end{proposition}

    \begin{proof}
        For the case of equity, let
        \[
            X = (r_{O}(a))_{a \in \fOrderedSeq{\sAgent}}, \quad
            Y = (q(a))_{a \in \fOrderedSeq{\sAgent}}.
        \]

        By construction, $X$ and $Y$ contain as many elements as $\sAgent$, and $|X| = |Y|$.
        Expanding $\fFair_{eqi'} (O)$ gives:
        \[
            [ \forall i \in \sNat, \ 1 \leq i \leq |X| \ \Rightarrow \
            \text{$i$-th elem. } (r_{O}(a))_{a \in \fOrderedSeq{\sAgent}} \geq
            \text{$i$-th elem. } (q(a))_{a \in \fOrderedSeq{\sAgent}} ] \ .
        \]

        This condition holds iff
        \[
            \forall a \in \sAgent \ (r_{O}(a) \geq q(a)),
        \]
        which is precisely the definition of $\fFair_{eqi} (O)$ in Definition~\ref{def:equality-equity}.
    \end{proof}

    \subsection{Implementation}

    The \Tiles framework is implemented as an open-source project\footnote{\url{https://github.com/julianmendez/tiles}} written in the \Soda language~\cite{Mendez-2023-Soda}, an open-source functional language\footnote{\url{https://github.com/julianmendez/soda}}.
    Code written in \Soda can be formally verified using the Lean~\cite{Lean-2013} proof assistant and seamlessly integrated into the Java Virtual Machine ecosystem, enabling efficient execution~\cite{Mendez.Kampik-2025-EUMAS}.
    The \Tiles implementation in \Soda aims to follow the pipeline notation as closely as possible, with each tile's source code directly accessible.

    The framework includes detailed components that ensure the correct construction of pipelines.
    These components are particularly focused on zipping and unzipping sequences, as well as creating and projecting tuples.
    Since \Soda is statically typed and \Tiles is a typed framework, type consistency across the entire pipeline can be verified at compile time.

    \subsection{Concurrent Execution}

    One of the advantages of the framework is its support for concurrent execution.
    However, there is an important caveat: the framework cannot guarantee preserving order, complexity, and concurrency at the same time.

    The reason is that the underlying data structures cannot maintain order during concurrent execution.
    Technically, it is possible to process a sequence by launching multiple threads on different cores, but once those threads finish, their results must be collected back into a sequence. For efficiency, this sequence is constructed in the order in which the threads finish, which may differ from the original order.

    A workaround is to pair each element with its original index before processing. After collecting the results, the sequence can then be sorted by these indices. However, this additional step may affect the overall complexity.

    Table~\ref{tab:primitive-tiles-concurrency} analyzes the primitive tiles and their potential for independent evaluation of elements within a sequence.

    \begin{longtable}{lll}
        \caption{Primitive tiles and concurrent evaluation} \\
        \hline
        \textbf{Function} & \textbf{Independent Evaluation}                          \\
        \hline
        &                                                          \\
        $\text{cross}((a_{i})_{i=1}^{n},(b_{j})_{j=1}^{m})$
        & \makecell[l]{Yes, each of the $n\cdot m$ pairs can be \\ constructed independently.}                  \\
        & \\
        $\text{zip}((a_{i})_{i=1}^{n},(b_{j})_{j=1}^{m})$
        & \makecell[l]{Yes, each of the $\min\{n,m\}$ pairs can be \\ constructed independently.} \\
        &                                                          \\
        $\text{apply}_\varphi(a)$
        & Yes, as it applies to a single element only.
        \\
        &                                                          \\
        $\text{bind}_\varphi((a_{i})_{i=1}^{n})$
        & \makecell[l]{Yes, each element can be mapped \\ independently.}
        \\
        &                                                          \\
        $\text{fold}_\varphi(z,(a_{i})_{i=1}^{n})$
        & \makecell[l]{No, evaluation of each element depends \\ on the results of previous ones.}                  \\
        &                                                          \\
        \hline
        \label{tab:primitive-tiles-concurrency}
    \end{longtable}

    \section{Discussion}
    \label{sec:discussion}

    This section informally discusses the capabilities and limitations of the \AR metamodel and the \Tiles framework.
    The framework is designed around the concept of \emph{flow}, where data moves through a pipeline.
    This pipeline connects tiles, forming a directed acyclic graph with a single start and end point.
    We do not provide an effective algorithm for constructing pipelines from the first-order formula defined by the fairness measure, since this is not possible in the general case.
    In our metamodel, we exclude functions that operate on multiple agents, multiple resources, or combinations thereof.

    \subsection{Expressiveness}

    The \AR metamodel can represent a wide range of fairness scenarios, though not all possible ones.
    We specifically focus on scenarios involving attributes of agents and resources.
    In particular, the metamodel supports the combination of multiple agents and resources, as demonstrated in Example~\ref{ex:preferences}, where an agent or resource serves as the first parameter of a function that returns a sequence.

    \subsection{Decidability and Time Complexity}

    Pipelines built within the \Tiles framework are decidable provided that their tiles are decidable.
    When contracts between tiles are respected, undefined values cannot arise.

    Regarding time complexity, the pipeline structure guarantees the absence of loops in the diagram, ensuring that no tile is evaluated more than once.
    The worst-case complexity of a pipeline is determined by the maximum complexity of its individual tiles.

    \subsection{Applicability}

    The \Tiles framework has broad applicability beyond the fairness domain.
    As a software engineering tool, it can handle any scenario involving finite, iterable sets of identifiers and attributes.
    The type system in \Tiles is flexible and includes sets of identifiers, quantities, Boolean values, tuples, and sequences.

    That said, the \Tiles framework is not intended to be applied at every level of a scenario.
    The framework describes connections and processes to provide clearer explanations of complex formulas.
    However, when formulas are sufficiently clear, they should be used directly instead of pipelines.

    The \Tiles framework is primarily intended for modeling, but the generated pipelines are naturally amenable to parallel execution.
    This is because pipelines are typically designed to process multiple agents, resources, or quantities simultaneously.
    Nevertheless, designing pipelines for parallel execution requires expertise, as it can be effort-intensive.

    \subsection{Limitations}

    A clear limitation of our approach lies in the generic nature of the metamodel.
    To address this, we introduced \Tiles as a layer of modular building blocks.
    Although \Tiles is conceptually elegant, it may, like other declarative notations, present challenges in terms of readability.
    Studying (and potentially improving) the readability of \Tiles is therefore an important direction for future work.
    For example, one could instantiate fairness measures and scenarios in several languages and systematically compare their comprehensibility, or conduct perceived usefulness studies analogous to~\cite{DBLP:journals/sosym/Jalali23}, which empirically compares several declarative modeling languages.

    Furthermore, this work does not address the modeling of \emph{specific} real-world problems.
    Instead, we either use \emph{toy examples} to facilitate understanding, or rely on generally applicable measures, for instance, to showcase broad applicability at the conceptual level or to highlight fundamental relationships (such as strict equity versus equality).
    Future work could evaluate our metamodel and its operationalization in concrete case studies, where real-world users face fairness modeling and analysis challenges.

    \section{Conclusion}
    \label{sec:conclusion}

    We introduced a metamodel of fairness and demonstrated its application through concrete examples.
    The metamodel provides a structured approach to understanding and evaluating fairness across diverse scenarios.
    By first identifying agents, resources, and their relevant attributes, users can construct nuanced fairness measures and instantiate scenarios from them.

    We presented examples illustrating central fairness concepts such as equality, equity, group fairness, and individual fairness, as well as continuous measures, such as the Gini and Theil indices.
    We explored applications of the metamodel ranging from economics and game theory (e.g., preferences and envy-freeness) to computer networks (e.g., Jain's fairness index) and examined a real-world case involving Child Care Subsidy in Australia.

    In addition to the examples, we proposed a method for operationalizing the metamodel that simplifies formal notation.
    The graphical and conceptual representation is supported by a modular design, and the resulting notation for defining functions can be applied across multiple domains.
    Moreover, we provided an implementation of this representation through the \Tiles framework.

    Looking ahead, the fairness framework can be extended and refined through applications in real-world scenarios.
    In this context, it is essential to develop further tooling that simplifies practical modeling and supports broader adoption of metamodel-based fairness analysis.

    \subsubsection*{Acknowledgements}
    This work was partially supported by the Wallenberg AI, Autonomous Systems and Software Program (WASP), funded by the Knut and Alice Wallenberg Foundation.

    \bibliographystyle{LNGAI}
    \bibliography{main}

@article{AlerTubella-2022,
    author = {Aler Tubella, Andrea and Barsotti, Flavia and Ko{\c{c}}er, R{\"u}ya G{\"o}khan and Mendez, Julian Alfredo},
    title = {{Ethical implications of fairness interventions: what might be hidden behind engineering choices?}},
    journal = {Ethics and Information Technology},
    year = {2022},
    month = {02},
    day = {28},
    volume = {24},
    number = {1},
    pages = {12},
    issn = {1572-8439},
    doi = {10.1007/s10676-022-09636-z},
    url = {https://doi.org/10.1007/s10676-022-09636-z}
}

@misc{Lean-2013,
    title = {{Lean}},
    author = {{Leonardo de Moura - Microsoft Research}},
    year = {2013},
    url = {https://github.com/leanprover/lean4}
}

@article{DBLP:journals/sosym/Jalali23,
    author = {Amin Jalali},
    title = {Evaluating user acceptance of knowledge-intensive business process
 modeling languages},
    journal = {Softw. Syst. Model.},
    volume = {22},
    number = {6},
    pages = {1803--1826},
    year = {2023},
    url = {https://doi.org/10.1007/s10270-023-01120-6},
    doi = {10.1007/S10270-023-01120-6},
    bibsource = {dblp computer science bibliography, https://dblp.org}
}

@inproceedings{Mendez.Kampik.Aler.Dignum-2024-SCAI,
    author = {Mendez, Julian Alfredo and Kampik, Timotheus and Aler Tubella, Andrea and Dignum, Virginia},
    title = { {A Clearer View on Fairness: Visual and Formal Representation for Comparative Analysis} },
    booktitle = {14th Scandinavian Conference on Artificial Intelligence, SCAI 2024},
    year = {2024},
    month = {06},
    pages = {112--120},
    editor = {Westphal, Florian and Peretz-Andersson, Einav and Riveiro, Maria and Bach, Kerstin and Heintz, Fredrik},
    organization = {Swedish Artificial Intelligence Society},
    doi = {10.3384/ecp208013},
    url = {https://ecp.ep.liu.se/index.php/sais/article/view/1005/913},
    note = { \url{https://doi.org/10.3384/ecp208013} }
}

@book{Weske-2019,
    author = {Weske, Mathias},
    title = {Business Process Management - Concepts, Languages, Architectures},
    publisher = {Springer Berlin, Heidelberg},
    doi = {10.1007/978-3-662-59432-2},
    url = {https://doi.org/10.1007/978-3-662-59432-2},
    isbn = {978-3-662-59431-5},
    edition = {3},
    year = {2019}
}

@book{Foley-1967-ResourceAllocation,
    author = {Duncan K. Foley},
    title = {Resource allocation and the public sector},
    publisher = {Yale Economics Essays},
    series = {},
    year = {1966}
}

@misc{Amanatidis-2018-ComparingEnvyFreeness,
    title = {Comparing Approximate Relaxations of Envy-Freeness},
    author = {Georgios Amanatidis and Georgios Birmpas and Evangelos Markakis},
    year = {2018},
    eprint = {1806.03114},
    archivePrefix = {arXiv},
    primaryClass = {cs.GT},
    doi = {10.48550/arXiv.1806.03114},
    url = {https://arxiv.org/abs/1806.03114},
}

@misc{Li-2024-PriceEnvyFreeness,
    title = {A Complete Landscape for the Price of Envy-Freeness},
    author = {Zihao Li and Shengxin Liu and Xinhang Lu and Biaoshuai Tao and Yichen Tao},
    year = {2024},
    eprint = {2401.01516},
    archivePrefix = {arXiv},
    primaryClass = {cs.GT},
    doi = {10.48550/arXiv.2401.01516},
    url = {https://arxiv.org/abs/2401.01516}
}

@inproceedings{Binns-2020-ConflictIndividualGroupFairness,
    author = {Binns, Reuben},
    title = {On the apparent conflict between individual and group fairness},
    year = {2020},
    isbn = {9781450369367},
    publisher = {Association for Computing Machinery},
    address = {New York, NY, USA},
    url = {https://doi.org/10.1145/3351095.3372864},
    doi = {10.1145/3351095.3372864},
    booktitle = {Proceedings of the 2020 Conference on Fairness, Accountability, and Transparency},
    pages = {514--524},
    numpages = {11},
    location = {Barcelona, Spain},
    series = {FAT* '20}
}

@article{Ramadan-2025-SSM,
    author = {Ramadan, Qusai and Konersmann, Marco and Ahmadian, Amir Shayan and J{\"u}rjens, Jan and Staab, Steffen},
    title = {MBFair: a model-based verification methodology for detecting violations of individual fairness},
    journal = {Software and Systems Modeling},
    year = {2025},
    month = {02},
    volume = {24},
    number = {1},
    pages = {111--136},
    doi = {10.1007/s10270-024-01184-y},
    url = {https://doi.org/10.1007/s10270-024-01184-y},
    issn = {1619-1374}
}

@misc{Mendez-2023-Soda,
    title = { {Soda: An Object-Oriented Functional Language for Specifying Human-Centered Problems} },
    author = {Mendez, Julian Alfredo},
    year = {2023},
    month = {10},
    eprint = {2310.01961},
    archivePrefix = {arXiv},
    doi = {10.48550/arXiv.2310.01961},
    url = {https://doi.org/10.48550/arXiv.2310.01961},
    note = { \url{https://doi.org/10.48550/arXiv.2310.01961} },
    primaryClass = {cs.PL}
}

@inproceedings{Mendez.Kampik-2025-EUMAS,
    author = {Mendez, Julian Alfredo and Kampik, Timotheus},
    title = { {Can Proof Assistants Verify Multi-agent Systems?} },
    booktitle = {Multi-Agent Systems},
    year = {2025},
    month = {06},
    pages = {323--339},
    editor = {Collier, Rem and Ricci, Alessandro and Nallur, Vivek and Burattini, Samuele and Omicini, Andrea},
    publisher = {Springer Nature Switzerland},
    address = {Cham},
    doi = {10.1007/978-3-031-93930-3_19},
    url = {https://doi.org/10.1007/978-3-031-93930-3_19},
    note = { \url{https://doi.org/10.1007/978-3-031-93930-3_19} }
}

@article{DBLP:journals/tvcg/WexlerPBWVW20,
    author = {James Wexler and
 Mahima Pushkarna and
 Tolga Bolukbasi and
 Martin Wattenberg and
 Fernanda B. Vi{\'{e}}gas and
 Jimbo Wilson},
    title = {{The What-If Tool: Interactive Probing of Machine Learning Models}},
    journal = {{IEEE} Trans. Vis. Comput. Graph.},
    volume = {26},
    number = {1},
    pages = {56--65},
    year = {2020},
    url = {https://doi.org/10.1109/TVCG.2019.2934619},
    doi = {10.1109/TVCG.2019.2934619},
    bibsource = {dblp computer science bibliography, https://dblp.org}
}

@article{DBLP:journals/ibmrd/BellamyDHHHKLMM19,
    author = {Rachel K. E. Bellamy and
 Kuntal Dey and
 Michael Hind and
 Samuel C. Hoffman and
 Stephanie Houde and
 Kalapriya Kannan and
 Pranay Lohia and
 Jacquelyn Martino and
 Sameep Mehta and
 Aleksandra Mojsilovic and
 Seema Nagar and
 Karthikeyan Natesan Ramamurthy and
 John T. Richards and
 Diptikalyan Saha and
 Prasanna Sattigeri and
 Moninder Singh and
 Kush R. Varshney and
 Yunfeng Zhang},
    title = {{AI Fairness 360: An extensible toolkit for detecting and mitigating algorithmic bias}},
    journal = {{IBM} J. Res. Dev.},
    volume = {63},
    number = {4/5},
    pages = {4:1--4:15},
    year = {2019},
    url = {https://doi.org/10.1147/JRD.2019.2942287},
    doi = {10.1147/JRD.2019.2942287},
    bibsource = {dblp computer science bibliography, https://dblp.org}
}

@article{bird2020fairlearn,
    title = {{Fairlearn: A toolkit for assessing and improving fairness in AI}},
    author = {Bird, Sarah and Dud{\'\i}k, Miro and Edgar, Richard and Horn, Brandon and Lutz, Roman and Milan, Vanessa and Sameki, Mehrnoosh and Wallach, Hanna and Walker, Kathleen},
    journal = {Microsoft, Tech. Rep. MSR-TR-2020-32},
    year = {2020}
}

@inproceedings{DBLP:conf/innovations/DworkHPRZ12,
    author = {Cynthia Dwork and
 Moritz Hardt and
 Toniann Pitassi and
 Omer Reingold and
 Richard S. Zemel},
    editor = {Shafi Goldwasser},
    title = {Fairness through awareness},
    booktitle = {Innovations in Theoretical Computer Science 2012, Cambridge, MA, USA,
 January 8-10, 2012},
    pages = {214--226},
    publisher = {{ACM}},
    year = {2012},
    url = {https://doi.org/10.1145/2090236.2090255},
    doi = {10.1145/2090236.2090255},
    bibsource = {dblp computer science bibliography, https://dblp.org}
}

@inproceedings{10.1145/3461702.3462621,
    author = {Fleisher, Will},
    title = {{What's Fair about Individual Fairness?}},
    year = {2021},
    isbn = {9781450384735},
    publisher = {Association for Computing Machinery},
    address = {New York, NY, USA},
    url = {https://doi.org/10.1145/3461702.3462621},
    doi = {10.1145/3461702.3462621},
    booktitle = {Proceedings of the 2021 AAAI/ACM Conference on AI, Ethics, and Society},
    pages = {480--490},
    numpages = {11},
    location = {Virtual Event, USA},
    series = {AIES '21}
}

@inproceedings{AlerTubella-2023,
    author = {Aler Tubella, Andrea and Coelho Mollo, Dimitri and Dahlgren Lindstr\"{o}m, Adam and Devinney, Hannah and Dignum, Virginia and Ericson, Petter and Jonsson, Anna and Kampik, Timotheus and Lenaerts, Tom and Mendez, Julian Alfredo and Nieves, Juan Carlos},
    title = {{ACROCPoLis: A Descriptive Framework for Making Sense of Fairness}},
    year = {2023},
    isbn = {9798400701924},
    publisher = {Association for Computing Machinery},
    address = {New York, NY, USA},
    url = {https://doi.org/10.1145/3593013.3594059},
    doi = {10.1145/3593013.3594059},
    booktitle = {Proceedings of the 2023 ACM Conference on Fairness, Accountability, and Transparency},
    pages = {1014--1025},
    numpages = {12},
    location = {Chicago, IL, USA},
    series = {FAccT '23}
}

@inproceedings{Mendez.Kampik-2025-LNGAI,
  author = {Mendez, Julian Alfredo and Kampik, Timotheus},
  title = {{Specification, Application, and Operationalization of a Metamodel of Fairness}},
  editor = {Liao, Beishui and Rotolo, Antonino and van der Torre, Leendert and Yu, Liuwen},
  booktitle = {{Logics for New-Generation AI 2025}},
  series = {{Logics for New-Generation AI}},
  volume = {5},
  month = {12},
  pages = {163--180},
  year = {2025},
  isbn = {978-1-84890-495-8},
  url = {https://www.collegepublications.co.uk/LNGAI/?00005},
}

@misc{Mendez.Kampik-2025-Specification-arxiv,
    title = {{Specification, Application, and Operationalization of a Metamodel of Fairness}},
    author = {Julian Alfredo Mendez and Timotheus Kampik},
    year = {2025},
    eprint = {2511.11144},
    archivePrefix = {arXiv},
    primaryClass = {cs.CY},
    url = {https://arxiv.org/abs/2511.11144},
    doi = {10.48550/arXiv.2511.11144}
}

@article{Richter.Rubinstein-2020-EconomicTheory,
    title = {The permissible and the forbidden},
    journal = {Journal of Economic Theory},
    volume = {188},
    pages = {105042},
    year = {2020},
    issn = {0022-0531},
    doi = {https://doi.org/10.1016/j.jet.2020.105042},
    url = {https://www.sciencedirect.com/science/article/pii/S0022053120300405},
    author = {Michael Richter and Ariel Rubinstein}
}

@misc{Jain-1984-QuantitativeMeasureOfFairness,
    author = {Jain, Rajendra K. and Chiu, Dah-Ming W. and Hawe, William R.},
    year = {1984},
    title = {{A Quantitative Measure of Fairness and Discrimination for Resource Allocation in Shared Computer Systems}},
    publisher = {DEC Research Report TR-301}
}

@misc{ProPublica-2016,
    author = {Larson, Jeff and Mattu, Surya and Kirchner, Lauren and Angwin, Julia},
    title = {{How We Analyzed the COMPAS Recidivism Algorithm}},
    year = {2016},
    publisher = {ProPublica},
    url = {https://www.propublica.org/article/how-we-analyzed-the-compas-recidivism-algorithm}
}

@article{Albarghouthi-2017,
    author = {Albarghouthi, Aws and D'Antoni, Loris and Drews, Samuel and Nori, Aditya V.},
    title = {FairSquare: Probabilistic Verification of Program Fairness},
    year = {2017},
    issue_date = {October 2017},
    publisher = {Association for Computing Machinery},
    address = {New York, NY, USA},
    volume = {1},
    number = {OOPSLA},
    note = {\url{https://doi.org/10.1145/3133904}},
    url = {https://doi.org/10.1145/3133904},
    doi = {10.1145/3133904},
    journal = {Proc. ACM Program. Lang.},
    month = {10},
    articleno = {80},
    numpages = {30}
}

@inproceedings{Albarghouthi-2019,
    author = {Albarghouthi, Aws and Vinitsky, Samuel},
    title = {Fairness-Aware Programming},
    year = {2019},
    isbn = {9781450361255},
    publisher = {Association for Computing Machinery},
    address = {New York, NY, USA},
    note = {\url{https://doi.org/10.1145/3287560.3287588}},
    url = {https://doi.org/10.1145/3287560.3287588},
    doi = {10.1145/3287560.3287588},
    booktitle = {Proceedings of the Conference on Fairness, Accountability, and Transparency},
    pages = {211--219},
    numpages = {9},
    location = {Atlanta, GA, USA},
    series = {FAT* '19}
}

@inproceedings{Dwork-2012,
    author = {Dwork, Cynthia and Hardt, Moritz and Pitassi, Toniann and Reingold, Omer and Zemel, Richard},
    title = {Fairness through Awareness},
    year = {2012},
    isbn = {9781450311151},
    publisher = {Association for Computing Machinery},
    address = {New York, NY, USA},
    note = {\url{https://doi.org/10.1145/2090236.2090255}},
    url = {https://doi.org/10.1145/2090236.2090255},
    doi = {10.1145/2090236.2090255},
    booktitle = {Proceedings of the 3rd Innovations in Theoretical Computer Science Conference},
    pages = {214--226},
    numpages = {13},
    location = {Cambridge, Massachusetts},
    series = {ITCS '12}
}

@inproceedings{Hardt-2016,
    archivePrefix = {arXiv},
    arxivId = {1610.02413},
    author = {Hardt, Moritz and Price, Eric and Srebro, Nathan},
    booktitle = {Advances in Neural Information Processing Systems},
    eprint = {1610.02413},
    issn = {10495258},
    pages = {3323--3331},
    title = {{Equality of opportunity in supervised learning}},
    year = {2016}
}

@inproceedings{Joseph-2016,
    title = {{Fairness in Learning: Classic and Contextual Bandits}},
    author = {Joseph, Matthew and Kearns, Michael and Morgenstern, Jamie H and Roth, Aaron},
    booktitle = {{Advances in Neural Information Processing Systems}},
    editor = {D. Lee and M. Sugiyama and U. Luxburg and I. Guyon and R. Garnett},
    pages = {325--333},
    publisher = {Curran Associates, Inc.},
    volume = {29},
    year = {2016},
    series = {NIPS},
    url = {https://proceedings.neurips.cc/paper_files/paper/2016/file/eb163727917cbba1eea208541a643e74-Paper.pdf}
}

@inproceedings{Kearns-2018,
    title = {Preventing fairness gerrymandering: Auditing and learning for subgroup fairness},
    author = {Kearns, Michael and Neel, Seth and Roth, Aaron and Wu, Zhiwei Steven},
    booktitle = {International Conference on Machine Learning},
    pages = {2564--2572},
    year = {2018},
    organization = {PMLR}
}

@article{CorbettDavies-2018,
    author = {Sam Corbett{-}Davies and Sharad Goel},
    title = {The Measure and Mismeasure of Fairness: {A} Critical Review of Fair Machine Learning},
    journal = {CoRR},
    volume = {abs/1808.00023},
    year = {2018},
    url = {http://arxiv.org/abs/1808.00023},
    eprinttype = {arXiv},
    eprint = {1808.00023},
    bibsource = {dblp computer science bibliography, https://dblp.org}
}

@article{Chouldechova-2017,
    author = {Chouldechova, Alexandra},
    title = {Fair Prediction with Disparate Impact: A Study of Bias in Recidivism Prediction Instruments},
    journal = {Big Data},
    volume = {5},
    number = {2},
    pages = {153-163},
    year = {2017},
    doi = {10.1089/big.2016.0047},
    note = {PMID: 28632438},
    url = {https://doi.org/10.1089/big.2016.0047},
    eprint = {https://doi.org/10.1089/big.2016.0047}
}

@misc{Binns-2019-IndividualGroupFairness,
    title = {On the Apparent Conflict Between Individual and Group Fairness},
    author = {Reuben Binns},
    year = {2019},
    month = {12},
    eprint = {1912.06883},
    archivePrefix = {arXiv},
    primaryClass = {cs.LG},
    url = {https://arxiv.org/abs/1912.06883},
    doi = {10.48550/arXiv.1912.06883}
}

@article{MussardSeyteTerraza-2003-GeneralizedEntropy,
    title = {Decomposition of Gini and the Generalized Entropy Inequality Measures},
    author = {Mussard, St\'{e}phane and Seyte, Fran\c{c}oise and Terraza, Michel},
    journal = {Economics Bulletin},
    volume = {4},
    number = {7},
    pages = {1--6},
    year = {2003},
    url = {https://accessecon.com/pubs/EB/2003/Volume4/EB-03D30001A.pdf},
    note = {Submitted: January 27, 2003. Accepted: January 27, 2003.}
}

@article{Dagum-1997-GeneralizedEntropyInequalityMeasures,
    title = {Decomposition and interpretation of Gini and the generalized entropy inequality measures},
    volume = {57},
    url = {https://rivista-statistica.unibo.it/article/view/1060},
    doi = {10.6092/issn.1973-2201/1060},
    number = {3},
    journal = {Statistica},
    author = {Dagum, Camilo},
    year = {1997},
    month = {01},
    pages = {295--308}
}

@article{Tsekouras_Tsallis-2005-GeneralizedEntropy,
    title = {Generalized entropy arising from a distribution of $q$ indices},
    author = {Tsekouras, G. A. and Tsallis, Constantino},
    journal = {Phys. Rev. E},
    volume = {71},
    issue = {4},
    pages = {046144},
    numpages = {8},
    year = {2005},
    month = {04},
    publisher = {American Physical Society},
    doi = {10.1103/PhysRevE.71.046144},
    url = {https://link.aps.org/doi/10.1103/PhysRevE.71.046144}
}

@book{Osborne.Rubinstein-2020-Models,
    author = {Osborne, Martin J. and Rubinstein, Ariel},
    title = {Models in Microeconomic Theory},
    publisher = {Open Book Publishers},
    year = {2020},
    doi = {10.11647/OBP.0211},
    url = {https://doi.org/10.11647/OBP.0211}
}

@article{Conceicao.Ferreira-2000-Theil,
    author = {Conceicao, Pedro and Ferreira, Pedro M.},
    title = {The Young Person's Guide to the Theil Index: Suggesting Intuitive Interpretations and Exploring Analytical Applications},
    journal = {SSRN Electronic Journal},
    year = {2000},
    doi = {10.2139/ssrn.228703},
    issn = {1556-5068},
    s2cid = {19009769}
}

@inproceedings{Kerkhoven-2025-Assessing,
    title = {Assessing machine learning fairness with multiple contextual norms},
    author = {Pim Kerkhoven and Virginia Dignum and Monowar Bhuyan},
    booktitle = {The 37th Benelux Conference on Artificial Intelligence and the 34th Belgian Dutch Conference on Machine Learning},
    year = {2025},
    pages = {1--18},
    url = {https://openreview.net/forum?id=X1lBR3RKRH}
}

    \clearpage
    \appendix

    \section{Appendix}
    \label{sec:appendix}

    In this appendix, we show examples of configurations for different fairness measures presented in this paper.

    \begin{figure*}[ht!]
        \centering
        \resizebox{\textwidth}{!}{
            \begin{tikzpicture}[x=1mm, y=1mm, box/.style={rectangle, draw, rounded corners=2mm, minimum width=4mm, minimum height=8mm, align=center}]
            \node [box] (S) {};
            \node [box, right=4mm of S] (allagent) {$\mathsf{\ all\text{-}agent \ \Big| \ _{(a)} \ accumulates \ \Big| \ _{(m)} \ distinct \ \Big| \ _{(m)}  \ length \ \Big| \ _{m} \ apply \ m=1 \ _{b}} $};
            \draw [->] (S.east) -- (allagent.west);
            \end{tikzpicture}
        }

        \caption{Pipeline for equality for a Child Care Subsidy using \Tiles.
        This uses generic tiles.}
        \label{fig:pipeline-equality-length}
    \end{figure*}
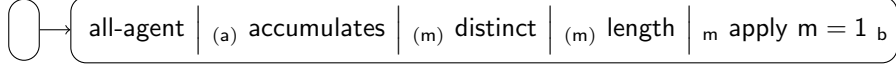

    \begin{figure*}[ht!]
        \centering
        \resizebox{\textwidth}{!}{
            \begin{tikzpicture}[x=1mm, y=1mm, scale=\pipelineScale, box/.style={rectangle, draw, rounded corners=2mm, minimum width=4mm, minimum height=8mm, align=center}]
            \node [box] (S) {};
            \node [box, right=6mm of S] (allagent0) {$\mathsf{\ all\text{-}agent \ \Big| \ _{(a)} \ accumulates \ _{(m_{0})}} $};
            \node [box, right=10mm of allagent0] (allagent1) {$\mathsf{\ all\text{-}agent \ \Big| \ _{(a)} \ needs \ _{(m_{1})}}$};
            \node [box, below=8mm of allagent0, xshift=8mm] (zip) {$\mathsf{_{(m_{0}), (m_{1})} \ zip \ \Big| \ _{(\langle m_{0}, m_{1}\rangle )} \ forall \ (m_{0} \geq m_{1}) \ _{b}} $};

            \draw [->] (S.east) -- ++(2mm, 0mm) |- (allagent0.west);
            \draw [->] (S.east) -- ++(2mm, 0mm) -- ++(0mm, 7mm) -- ++(59mm, 0mm) |-  (allagent1.west);
            \draw [->] (allagent0.east) -- ++(2mm, 0mm) -| ++(0mm, -7mm) -- ++(-56mm, 0mm) |- (zip.west);
            \draw [->] (allagent1.east) -- ++(2mm, 0mm) -| ++(0mm, -9mm) -- ++(-103mm, 0mm) |- (zip.west);
            \end{tikzpicture}
        }
        \caption{Pipeline for equity for a Child Care Subsidy using \Tiles.
        This diagram uses \tilesfun{filter}, \tilesfun{length}, and \tilesfun{apply}.}
        \label{fig:pipeline-equity-length}
    \end{figure*}
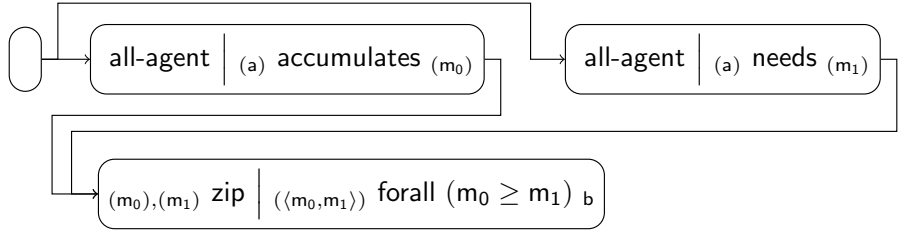

    \begin{figure*}[ht!]
        \centering
        \resizebox{\textwidth}{!}{
            \begin{tikzpicture}[x=1mm, y=1mm, scale=\pipelineScale, box/.style={rectangle, draw, rounded corners=2mm, minimum width=4mm, minimum height=8mm, align=center}]

            \node[box] (S) {};
            \node[box, right=10mm of S] (allagent0) {$\mathsf{all\text{-}agent \ _{(a)}}$};
            \node[box, below=4mm of allagent0, xshift=2mm] (allresource0) {$\mathsf{all\text{-}resource \ _{(r)}}$};
            \node[box, right=8mm of allagent0] (filter0) {$\mathsf{_{(a),(r) } \ cross \ \Big| \ _{(\langle a, r \rangle)} \ filter \ \mathnormal{p}(a) \land a \text{ receives } r \ \Big| \ _{(\langle a, r \rangle)} \ length \ _{m_{0}}}$};
            \node[box, below=6mm of allresource0, xshift=18mm] (allagent1) {$\mathsf{all\text{-}agent \ \Big| \ _{(a)}  \ filter \ \mathnormal{p}(a) \ \Big| \ _{(a)} \ length \ _{m_{1}}}$};
            \node[box, right=6mm of allagent1] (rpositive) {$\mathsf{_{m_{0}, m_{1}} \ apply \ \frac{m_{0}}{m_{1}} \ _{m_{2}}}$};
            \node[box, below=20mm of allresource0, xshift=-2mm] (allagent2) {$\mathsf{all\text{-}agent \ _{(a)}}$};
            \node[box, below=4mm of allagent2, xshift=2mm] (allresource1) {$\mathsf{all\text{-}resource \ _{(r)}}$};
            \node[box, right=8mm of allagent2] (filter1) {$\mathsf{_{(a),(r) } \ cross \ \Big| \ _{(\langle a, r \rangle)} \ filter \ \lnot \mathnormal{p}(a) \land a \text{ receives } r \ \Big| \ _{(\langle a, r \rangle)} \ length \ _{m_{0}}}$};
            \node[box, below=6mm of allresource1, xshift=19mm] (allagent3) {$\mathsf{all\text{-}agent \ \Big| \ _{(a)}  \ filter \ \lnot \mathnormal{p}(a) \ \Big| \ _{(a)} \ length \ _{m_{1}}}$};
            \node[box, right=6mm of allagent3] (rnegative) {$\mathsf{_{m_{0}, m_{1}} \ apply \ \frac{m_{0}}{m_{1}} \ _{m_{3}}}$};
            \node[box, below=6mm of rnegative] (result) {$\mathsf{_{m_{2}, m_{3}} \ apply \ m_{2} \simeq_{\varepsilon} m_{3} \ _{b}}$};

            \draw [->] (S.east) -- (allagent0.west);
            \draw [->] (S.east) -- ++(3mm, 0mm) |- (allagent1.west);
            \draw [->] (S.east) -- ++(3mm, 0mm) |- (allagent2.west);
            \draw [->] (S.east) -- ++(3mm, 0mm) |- (allagent3.west);
            \draw [->] (S.east) -- ++(3mm, 0mm) |- (allresource0.west);
            \draw [->] (S.east) -- ++(3mm, 0mm) |- (allresource1.west);
            \draw [->] (allagent0.east) -- (filter0.west);
            \draw [->] (allresource0.east) -- ++(2mm, 0mm) |- (filter0.west);
            \draw [->] (allagent2.east) -- (filter1.west);
            \draw [->] (allresource1.east) -- ++(2mm, 0mm) |- (filter1.west);
            \draw [->] (allagent1.east) -- (rpositive.west);
            \draw [->] (allagent3.east) -- (rnegative.west);
            \draw [->] (filter0.east) -- ++(2mm, 0mm) -| ++(0mm, -18mm) -- ++(-60mm, 0mm) |- (rpositive.west);
            \draw [->] (filter1.east) -- ++(2mm, 0mm) -| ++(0mm, -18mm) -- ++(-60mm, 0mm) |- (rnegative.west);
            \draw [->] (rpositive.east) -- ++(33mm, 0mm) -| ++(0mm, -48mm) -- ++(-68mm, 0mm) |- (result.west);
            \draw [->] (rnegative.east) -- ++(2mm, 0mm) -| ++(0mm, -6mm) -- ++(-42mm, 0mm) |- (result.west);
            \end{tikzpicture}
        }
        \caption{Pipeline for group fairness with respect to a protected attribute $p$.}
        \label{fig:pipeline-group-fairness}
    \end{figure*}
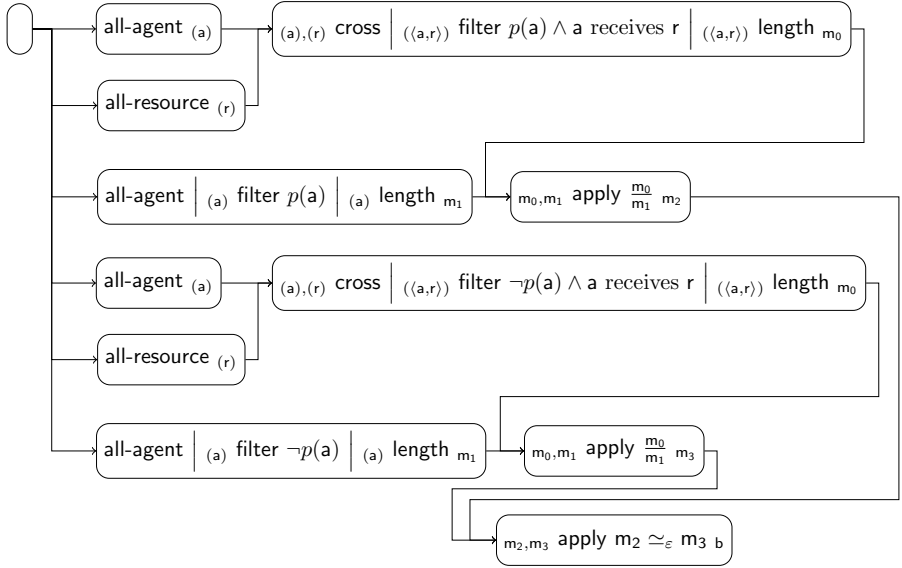

    \begin{figure*}[ht!]
        \centering
        \resizebox{\textwidth}{!}{
            \begin{tikzpicture}[x=1mm, y=1mm, scale=\pipelineScale, box/.style={rectangle, draw, rounded corners=2mm, minimum width=4mm, minimum height=8mm, align=center}]

            \node[box] (S) {};
            \node[box, right=6mm of S] (allagent0) {$\mathsf{all\text{-}agent \ _{(a_{0})}}$};
            \node[box, below=4mm of allagent0] (allagent1) {$\mathsf{all\text{-}agent \ _{(a_{1})}}$};
            \node[box, below=4mm of allagent1, xshift=1mm] (allresource) {$\mathsf{all\text{-}resource \ _{(r)}}$};
            \node[box, right=4mm of allagent0] (filter0) {$\mathsf{_{(a_{0}), (a_{1})} \ cross \ \Big| \ _{(\langle a_{0}, a_{1} \rangle)} \ filter \ a_{0} \neq a_{1} \land \mathnormal{q}(a_{0}) = \mathnormal{q}(a_{1})\ _{(\langle a_{0}, a_{1} \rangle)} }$};
            \node[box, below=30mm of filter0, xshift=-15mm] (filter1) {$\mathsf{_{(\langle a_{0}, a_{1}\rangle), (r)} \ cross \ \Big| \ _{(\langle \langle a_{0}, a_{1} \rangle , r \rangle)} \ forall \ \left( \makecell[l]{\mathsf{a_{0} \text{ receives } r \land a_{1} \text{ receive } r \ \lor} \\ \mathsf{\lnot (a_{0} \text{ receive } r) \land \lnot (a_{1} \text{ receives } r) } } \right) \ _{b} }$};

            \draw [->] (S.east) -- (allagent0.west);
            \draw [->] (S.east) -- ++(3mm, 0mm) |- (allagent1.west);
            \draw [->] (S.east) -- ++(3mm, 0mm) |- (allresource.west);
            \draw [->] (allagent0.east) -- (filter0.west);
            \draw [->] (filter0.east) -- ++(2mm, 0mm) -- ++(0mm, -32mm) -- ++(-119mm, 0mm) |- (filter1.west);
            \draw [->] (allagent1.east)  -- ++(2mm, 0mm) |- (filter0.west);
            \draw [->] (allresource.east) -- ++(2mm, 0mm) -- ++(0mm, -6mm) -- ++(-30mm, 0mm) |- (filter1.west);
            \end{tikzpicture}
        }
        \caption{Pipeline for individual fairness, considering an essential attribute $q$.}
        \label{fig:pipeline-individual-fairness}
    \end{figure*}
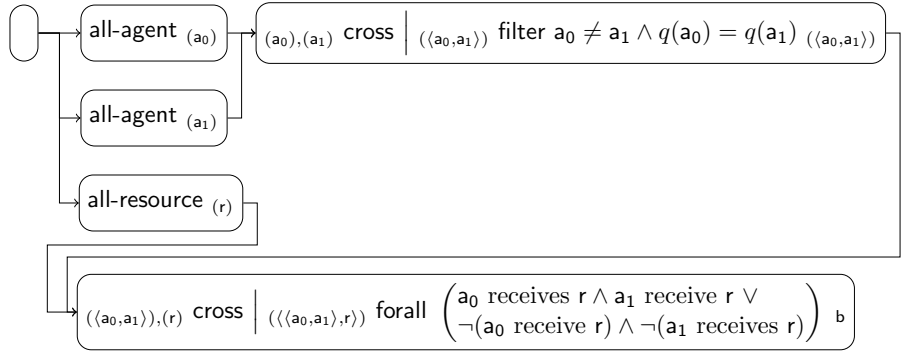

    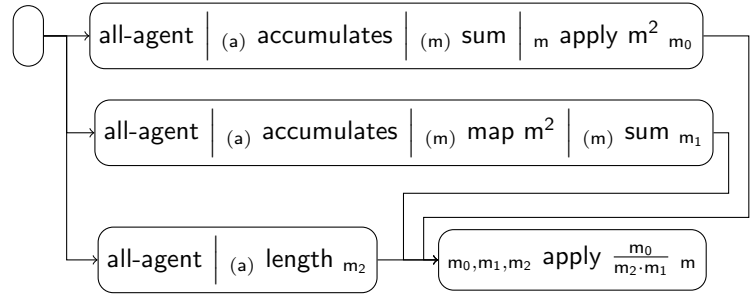
\begin{figure*}[ht!]
        \centering
        \begin{tikzpicture}[x=1mm, y=1mm, scale=\pipelineScale, box/.style={rectangle, draw, rounded corners=2mm, minimum width=4mm, minimum height=8mm, align=center}]

            \node[box] (S) {};
            \node[box, right=6mm of S] (allagent0) {$\mathsf{all\text{-}agent \ \Big| \ _{(a)} \ accumulates \ \Big| \ _{(m)} \ sum \ \Big| \ _{m} \ apply \ m^{2} \ _{m_{0}} }$};
            \node[box, below=4mm of allagent0, xshift=1mm] (allagent1) {$\mathsf{all\text{-}agent \ \Big| \ _{(a)} \ accumulates \ \Big| \ _{(m)} \ map \ m^{2} \ \Big| \ _{(m)} \ sum \ _{m_{1}} }$};
            \node[box, below=8mm of allagent1, xshift=-22mm] (allagent2) {$\mathsf{all\text{-}agent \ \Big| \ _{(a)} \ length \ _{m_{2}} }$};
            \node[box, right=8mm of allagent2] (apply) {$\mathsf{_{m_{0}, m_{1}, m_{2}} \ apply \ \frac{m_{0}}{m_{2} \cdot m_{1}} \ _{m}}$};

            \draw [->] (S.east) -- (allagent0.west);
            \draw [->] (S.east) -- ++(3mm, 0mm) |- (allagent1.west);
            \draw [->] (S.east) -- ++(3mm, 0mm) |- (allagent2.west);
            \draw [->] (allagent0.east) -- ++(6mm, 0mm) -- ++(0mm, -24mm) -- ++(-43mm, 0mm) |- (apply.west);
            \draw [->] (allagent1.east) -- ++(2mm, 0mm) -- ++(0mm, -8mm) -- ++(-43mm, 0mm)  |- (apply.west);
            \draw [->] (allagent2.east) -- (apply.west);
        \end{tikzpicture}
        \caption{Pipeline for Jain's index.}
        \label{fig:pipeline-jain}
    \end{figure*}

    \begin{figure*}[ht!]
        \centering
        \begin{tikzpicture}[x=1mm, y=1mm, scale=\pipelineScale, box/.style={rectangle, draw, rounded corners=2mm, minimum width=4mm, minimum height=8mm, align=center}]

            \node[box] (S) {};
            \node[box, right=6mm of S] (allagent0) {$\mathsf{all\text{-}agent \ \Big| \ _{(a)} \ accumulates \ _{(m_{0})} }$};
            \node[box, below=4mm of allagent0] (allagent1) {$\mathsf{all\text{-}agent \ \Big| \ _{(a)} \ accumulates \ _{(m_{1})} }$};
            \node[box, below=6mm of allagent1, xshift=21mm] (cross) {$\mathsf{ \ _{(m_{0}),(m_{1})} \ cross \ \Big| \ _{(\langle m_{0}, m_{1} \rangle)} \ map \ |m_{0} - m_{1}| \ \Big| \ _{(m_{2})} \ sum \ _{m_{3}} }$};
            \node[box, below=20mm of allagent1, xshift=7mm] (allagent2) {$\mathsf{all\text{-}agent \ \Big| \ _{(a)} \ accumulates \ \Big| \ _{(m)} \ sum \ _{m_{4}} }$};
            \node[box, below=8mm of allagent2, xshift=-12mm] (allagent3) {$\mathsf{all\text{-}agent \ \Big| \ _{(a)} \ length \ _{m_{5}} }$};
            \node[box, right=8mm of allagent3] (apply) {$\mathsf{_{m_{3}, m_{4}, m_{5}} \ apply \ 1 - \frac{m_{3}}{2 \cdot m_{5} \cdot m_{4}} \ _{m}}$};

            \draw [->] (S.east) -- (allagent0.west);
            \draw [->] (allagent0.east) -- ++(5mm, 0mm) -- ++(0mm, -21mm) -- ++(-54mm, 0mm) |- (cross.west);
            \draw [->] (allagent1.east) -- ++(3mm, 0mm) -- ++(0mm, -6mm) -- ++(-53mm, 0mm) |- (cross.west);
            \draw [->] (S.east) -- ++(3mm, 0mm) |- (allagent1.west);
            \draw [->] (S.east) -- ++(3mm, 0mm) |- (allagent2.west);
            \draw [->] (S.east) -- ++(3mm, 0mm) |- (allagent3.west);
            \draw [->] (cross.east) -- ++(4mm, 0mm) -- ++(0mm, -24mm) -- ++(-48mm, 0mm) |- (apply.west);
            \draw [->] (allagent2.east) -- ++(2mm, 0mm) -- ++(0mm, -8mm) -- ++(-23mm, 0mm)  |- (apply.west);
            \draw [->] (allagent3.east) -- (apply.west);
        \end{tikzpicture}
        \caption{Pipeline for the complement of the Gini index.
        Unlike the Gini index, 1 represents complete fairness and 0 represents complete unfairness.}
        \label{fig:pipeline-gini-index}
    \end{figure*}
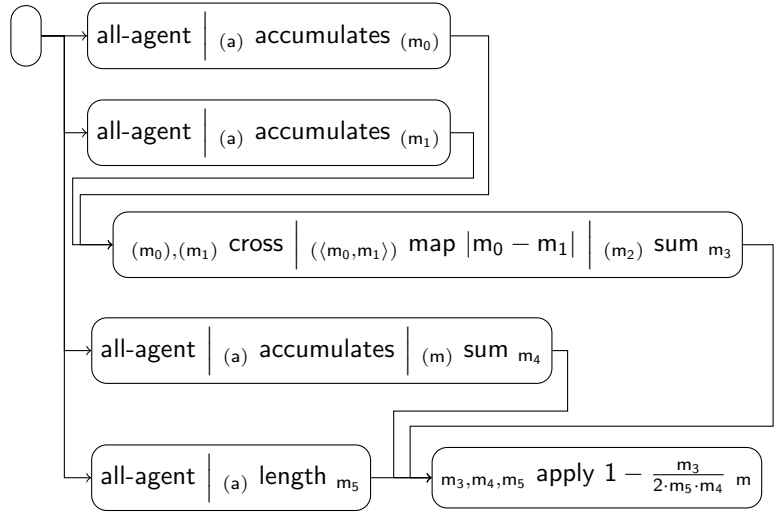

    \begin{figure*}[ht!]
        \centering
        \resizebox{\textwidth}{!}{
            \begin{tikzpicture}[x=1mm, y=1mm, scale=\pipelineScale, box/.style={rectangle, draw, rounded corners=2mm, minimum width=4mm, minimum height=8mm, align=center}]

            \node[box] (S) {};
            \node[box, right=8mm of S] (allagent0) {$\mathsf{all\text{-}agent \ \Big| \ _{(a)} \ accumulates \ \Big| \ _{(m)} \ sum \ _{m_{1}} }$};
            \node[box, below=6mm of allagent0, xshift=-12mm] (allagent1) {$\mathsf{all\text{-}agent \ \Big| \ _{(a)} \ length \ _{m_{2}} }$};
            \node[box, below=6mm of allagent1, xshift=6mm] (allagent2) {$\mathsf{all\text{-}agent \ \Big| \ _{(a)} \ accumulates \ _{(m_{0})} }$};
            \node[box, below=8mm of allagent2, xshift=21mm] (map) {$\mathsf{ \ _{(m_{0}); m_{1}, m_{2}} \ map \ \frac{m_{0}\cdot m_{2}}{m_{1}} \cdot \ln (\frac{m_{0}\cdot m_{2}}{m_{1}}) \ \Big| \ _{(m_{3});m_{2}} \ sum \ _{m_{4}} \ \Big| \ _{m_{4}, m_{2}} \ apply \ (1 - \frac{m_{4}}{m_{2}}) \ _{m} } $};

            \draw [->] (S.east) -- ++(3mm, 0mm) |- (allagent0.west);
            \draw [->] (allagent0.east) -- ++(5mm, 0mm) -- ++(0mm, -39mm) -- ++(-80mm, 0mm) |- (map.west);
            \draw [->] (allagent1.east) -- ++(25mm, 0mm) -- ++(0mm, -22mm) -- ++(-78mm, 0mm) |- (map.west);
            \draw [->] (allagent2.east) -- ++(3mm, 0mm) -- ++(0mm, -6mm) -- ++(-69mm, 0mm) |- (map.west);
            \draw [->] (S.east) -- ++(3mm, 0mm) |- (allagent1.west);
            \draw [->] (S.east) -- ++(3mm, 0mm) |- (allagent2.west);

            \end{tikzpicture}
        }
        \caption{Pipeline for the complement of the Theil index, where 1 represents complete fairness and 0 represents complete unfairness. Notice the notation with semicolon in the parameters. The parameters $\mathsf{m_{1}}$ and $\mathsf{m_{2}}$ are constant with respect to the $\mathsf{map}$ tile, which uses them. The parameter $\mathsf{m_{2}}$ is not used by the $\mathsf{sum}$ tile, but it is passed to the following tile. In both cases, the parameters after the semicolon are not part of the basic parameters of the tiles, but remain constant, regardless of whether they are used by the tile.}
        \label{fig:pipeline-theil-index}
    \end{figure*}
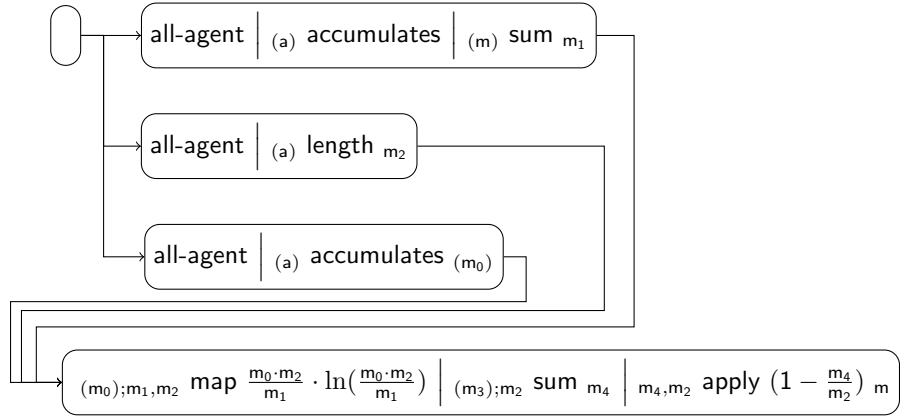

    \begin{figure*}[ht!]
        \centering
        \resizebox{\textwidth}{!}{
            \begin{tikzpicture}[x=1mm, y=1mm, scale=\pipelineScale, box/.style={rectangle, draw, rounded corners=2mm, minimum width=4mm, minimum height=8mm, align=center}]

            \node[box] (S) {};
            \node[box, right=6mm of S] (allagent0) {$\mathsf{all\text{-}agent \ \Big| \ _{(a)} \ map \ \mathnormal{[\langle \mathsf{a}, R_{\mathrm{high}} \rangle \in O \land res (\mathsf{a}) \neq R_{\mathrm{high}}]} \ _{(m_{0})} }$};
            \node[box, below=4mm of allagent0, xshift=-20mm] (allagent1) {$\mathsf{all\text{-}agent \ \Big| \ _{(a)} \ map \ [ \mathnormal{p} (a) ] \ _{(m_{1})}} $};

            \node[box, right=6mm of allagent1] (corr) {$\mathsf{_{(m_{0}), (m_{1})} \ correlation \ \Big| \ _{m_{2}} \ \mathsf{apply} \ |m_{2}| \ _{m}}$};

            \draw [->] (S.east) -- (allagent0.west);
            \draw [->] (S.east) -- ++(3mm, 0mm) |- (allagent1.west);
            \draw [->] (allagent0.east) -- ++(4mm, 0mm) -- ++(0mm, -6mm) -- ++(-40mm, 0mm) |- (corr.west);
            \draw [->] (allagent1.east) -- (corr.west);

            \end{tikzpicture}
        }
        \caption{Pipeline for the detection of false positives in a prediction system, where $pred$ is the prediction, $res$ is the actual fact, $R_{\mathrm{high}}$ represents the value to detect positiveness, and $p$ is a protected attribute.}
        \label{fig:pipeline-prediction}
    \end{figure*}
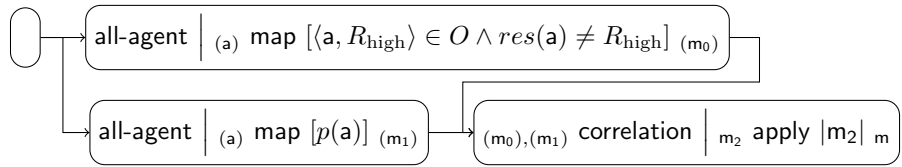

    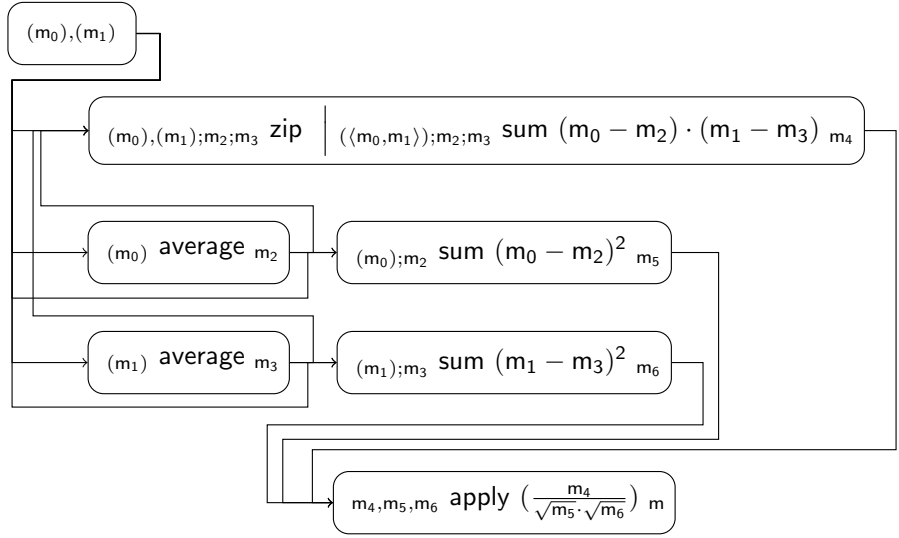
\begin{figure*}[ht!]
        \centering
        \resizebox{\textwidth}{!}{
            \begin{tikzpicture}[x=1mm, y=1mm, scale=\pipelineScale, box/.style={rectangle, draw, rounded corners=2mm, minimum width=4mm, minimum height=8mm, align=center}]

            \node[box] (S) {$\mathsf{ \ _{(m_{0}),(m_{1})} \ }$};
            \node[box, below=4mm of S, xshift=52mm] (sum0) {$\mathsf{ \ _{(m_{0}),(m_{1}); m_{2} ; m_{3}} \ zip \  \ \Big| \ _{(\langle m_{0}, m_{1} \rangle); m_{2}; m_{3}} \ sum \ (m_{0} - m_{2}) \cdot (m_{1} - m_{3}) \ _{m_{4}} }$};
            \node[box, below=20mm of S, xshift=15mm] (average0) {$\mathsf{ \ _{(m_{0})} \ average \ _{m_{2}} }$};
            \node[box, right=6mm of average0] (sum1) {$\mathsf{ \ _{(m_{0});m_{2}} \ sum \ (m_{0} - m_{2})^{2} \ _{m_{5}} }$};
            \node[box, below=6mm of average0] (average1) {$\mathsf{ \ _{(m_{1})} \ average \ _{m_{3}} }$};
            \node[box, right=6mm of average1] (sum2) {$\mathsf{ \ _{(m_{1});m_{3}} \ sum \ (m_{1} - m_{3})^{2} \ _{m_{6}} }$};
            \node[box, below=24mm of sum1] (apply) {$\mathsf{ \ _{m_{4}, m_{5}, m_{6}} \ apply \ (\frac{m_{4}}{\sqrt{m_{5}} \cdot \sqrt{m_{6}}}) \ _{m} }$};

            \draw [->] (S.east) -- ++(3mm, 0mm) -- ++(0mm, -6mm) -- ++(-19mm, 0mm) |- (sum0.west);
            \draw [->] (S.east) -- ++(3mm, 0mm) -- ++(0mm, -6mm) -- ++(-19mm, 0mm) -- ++(0mm, -28mm) -- ++(38mm, 0mm) |- (sum1.west);
            \draw [->] (S.east) -- ++(3mm, 0mm) -- ++(0mm, -6mm) -- ++(-19mm, 0mm) -- ++(0mm, -42mm) -- ++(38mm, 0mm) |- (sum2.west);
            \draw [->] (S.east) -- ++(3mm, 0mm) -- ++(0mm, -6mm) -- ++(-19mm, 0mm) |- (average0.west);
            \draw [->] (S.east) -- ++(3mm, 0mm) -- ++(0mm, -6mm) -- ++(-19mm, 0mm) |- (average1.west);
            \draw [->] (average0.east) -- ++(3mm, 0mm) -- ++(0mm, 6mm) -- ++(-35mm, 0mm) |- (sum0.west);
            \draw [->] (average0.east) -- (sum1.west);
            \draw [->] (average1.east) -- ++(3mm, 0mm) -- ++(0mm, 6mm) -- ++(-36mm, 0mm) |- (sum0.west);
            \draw [->] (average1.east) -- (sum2.west);
            \draw [->] (sum0.east) -- ++(4mm, 0mm) -- ++(0mm, -41mm) -- ++(-75mm, 0mm) |- (apply.west);
            \draw [->] (sum1.east) -- ++(6mm, 0mm) -- ++(0mm, -24mm) -- ++(-56mm, 0mm) |- (apply.west);
            \draw [->] (sum2.east)-- ++(4mm, 0mm) -- ++(0mm, -8mm) -- ++(-56mm, 0mm) |- (apply.west);
            \end{tikzpicture}
        }
        \caption{Component that computes the correlation, where the sequences $\mathsf{(m_{0}) \text{ and } (m_{1})}$ are the parameters of the component.
        The averages of $\mathsf{(m_{0}) \text{ and } (m_{1})}$ are stored in $\mathsf{m_{2}}$ and $\mathsf{m_{3}}$ respectively, which are computed once but used twice.
        The \tilesfun{zip} tile does not need $\mathsf{m_{2}}$ and $\mathsf{m_{3}}$, but it accepts them with the semicolon notation to pass them on to the \tilesfun{sum} tile.}
        \label{fig:pipeline-correlation}
    \end{figure*}

    \begin{figure*}[ht!]
        \centering
        \begin{tikzpicture}[x=1mm, y=1mm, box/.style={rectangle, draw, rounded corners=2mm, minimum width=4mm, minimum height=8mm, align=center}]

            \node [box] (S) {};
            \node [box, right=4mm of S] (allagent) {$\mathsf{\ all\text{-}agent \ \Big| \ _{(a)} \ accumulates \ \Big| \ _{(m)} \ forall \ m = 0 \ _{b}} $};
            \draw [->] (S.east) -- (allagent.west);

        \end{tikzpicture}
        \caption{Pipeline for no subsidy.
        The tile on the left provides all agents.
        The tile in the middle computes how much resource each agent received.
        The tile on the right checks that all resources are equal to 0.}
        \label{fig:no-subsidy}
    \end{figure*}
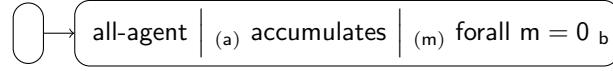

    \begin{figure*}[ht!]
        \centering
        \begin{tikzpicture}[x=1mm, y=1mm, scale=\pipelineScale, box/.style={rectangle, draw, rounded corners=2mm, minimum width=4mm, minimum height=8mm, align=center}]

            \node [box] (S) {};
            \node[box, right=6mm of S] (accumulates) {$\mathsf{all\text{-}agent \ \Big| \ _{(a)} \ accumulates \ _{(m_{0})}}$};
            \node[box, below=6mm of accumulates, xshift=4mm] (children)  {$\mathsf{all\text{-}agent \ \Big| \ _{(a)} \ map \ \mathnormal{children}(a) \ _{(m_{1})}}$};
            \node[box, below=22mm of accumulates, xshift=13mm] (allequal) {$\mathsf{_{(m_{0}), (m_{1})} \ zip \ \Big| \ _{(\langle m_{0}, m_{1} \rangle )} \  map \ \frac{m_{0}}{m_{1}} \ \Big| \ _{(m)} \ all\text{-}equal \ _{b} }$};
            \draw [->] (S.east)  -- (accumulates.west);
            \draw [->] (S.east) -- ++(3mm, 0mm) |- (children.west);
            \draw [->] (accumulates.east) -- ++(13mm, 0mm) -- ++(0mm, -23mm) -- ++(-63mm, 0mm) |- (allequal.west);
            \draw [->] (children.east) -- ++(3mm, 0mm) -- ++(0mm, -7mm) -- ++(-63mm, 0mm) |- (allequal.west);

        \end{tikzpicture}
        \caption{Representation of ``per child'' using \Tiles.
        The tile on the left provides agents, which are divided in two branches.
        The upper branch computes how much each agent (a family) has received and the lower branch how many children the family has.
        Both values are zipped back to compute the division.
        Note that we assume that each family has at least a child, but otherwise, if the number of children is 0, the division would be computed as undefined.}
        \label{fig:per-child}
    \end{figure*}
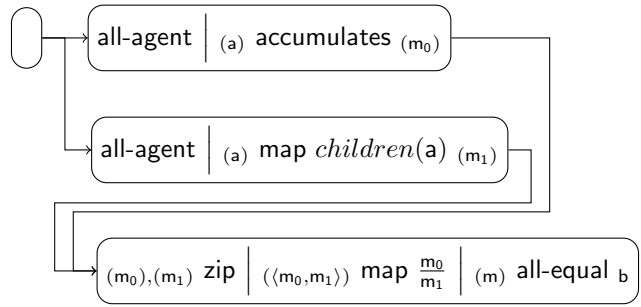

    \begin{figure*}[ht!]
        \centering
        \begin{tikzpicture}[x=1mm, y=1mm, box/.style={rectangle, draw, rounded corners=2mm, minimum width=4mm, minimum height=8mm, align=center}]

            \node [box] (S) {};
            \node [box, right=4mm of S] (allagent) {$\mathsf{\ all\text{-}agent \ \Big| \ _{(a)} \ accumulates \ \Big| \ _{(m)} \ all\text{-}equal \ _{b}} $};
            \draw [->] (S.east) -- (allagent.west);

        \end{tikzpicture}
        \caption{Representation of ``per family'' using \Tiles.
        This is equivalent to a standard equality pipeline where each agent receives exactly the same amount of resource.}
        \label{fig:per-family}
    \end{figure*}
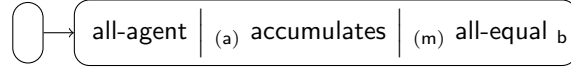

    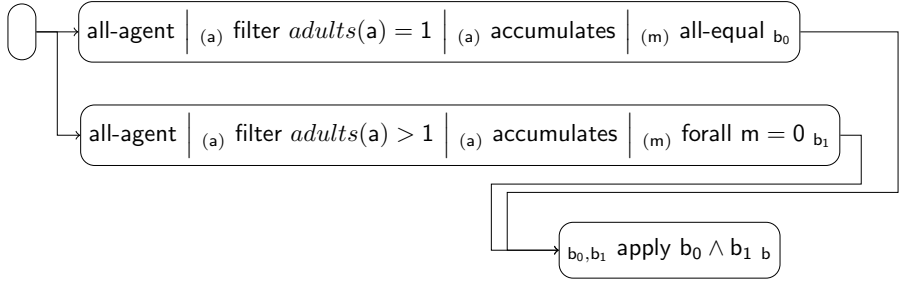
\begin{figure*}[ht!]
        \centering
        \resizebox{\textwidth}{!}{
            \begin{tikzpicture}[x=1mm, y=1mm, scale=\pipelineScale, box/.style={rectangle, draw, rounded corners=2mm, minimum width=4mm, minimum height=8mm, align=center}]

            \node [box] (S) {};
            \node[box, right=6mm of S] (single) {$\mathsf{all\text{-}agent \ \Big| \ _{(a)} \ filter \ \mathnormal{adults} (a) = 1 \ \Big| \ _{(a)} \ accumulates \ \Big| \ _{(m)} \ all\text{-}equal \ _{b_{0}}}$};
            \node[box, below=6mm of single, xshift=3mm] (multiple) {$\mathsf{all\text{-}agent \ \Big| \ _{(a)} \ filter \ \mathnormal{adults} (a) > 1 \ \Big| \ _{(a)} \ accumulates \ \Big| \ _{(m)} \ forall \ m = 0 \ _{b_{1}}}$};
            \node[box, below=8mm of multiple, xshift=30mm] (apply) {$\mathsf{_{b_{0}, b_{1}} \  apply \ b_{0} \land b_{1} \ _{b} }$};
            \draw [->] (S.east)  -- (single.west);
            \draw [->] (S.east) -- ++(3mm, 0mm) |- (multiple.west);
            \draw [->] (single.east) -- ++(14mm, 0mm) -- ++(0mm, -23mm) -- ++(-56mm, 0mm) |- (apply.west);
            \draw [->] (multiple.east) -- ++(3mm, 0mm) -- ++(0mm, -7mm) -- ++(-53mm, 0mm) |- (apply.west);
        \end{tikzpicture}
        }
        \caption{Representation of ``single guardian'' using \Tiles.
        This pipeline has two main branches.
        The upper branch accepts only families with one adult, i.e., single-parent/guardian families.
        The lower branch accepts all remaining families.
        It is worth noting that the sequences in both branches may have different number of elements and cannot be zipped back.
        On the other hand, the Boolean computation is combined with the `and' tile, on the right.}
        \label{fig:single-guardian}
    \end{figure*}

\end{document}